%% file: main.tex
\documentclass[a4paper,twocolumn,11pt,accepted=2025-10-20]{quantumarticle}

\pdfoutput=1
 
\usepackage[utf8]{inputenc}
\usepackage{float}
\usepackage{amsmath}
\usepackage{amsthm}
\usepackage{appendix}
\usepackage{braket}
\usepackage{physics}
\usepackage{bm}
\usepackage{bbm}
\usepackage[numbers,sort&compress]{natbib}
\usepackage{graphicx}
\usepackage{svg}
\usepackage{tcolorbox}
\usepackage{tikz}
\usepackage{tikzit}
\usepackage{mathtools}
\usepackage{amsfonts}
\input{colour_code_circuits/color_code.tikzstyles}
\usepackage{hyperref}
\usepackage[export]{adjustbox}

\usetikzlibrary{quantikz}
\usetikzlibrary{arrows, arrows.meta}

\theoremstyle{definition}
\newtheorem{definition}{Definition}[section]
\newtheorem{theorem}{Theorem}
\newtheorem{lemma}[theorem]{Lemma}
\newtheorem{example}[theorem]{Example}

\newtheorem{corollary}[theorem]{Corollary}

\newcommand{\zz}{\mathbbm{Z}}
\newcommand{\triangleqed}{
  \begingroup
  \renewcommand{\qedsymbol}{\ensuremath{\triangleleft}}
  \qed
  \endgroup
}

\begin{document}

\title{Designing fault-tolerant circuits using detector error models}

\author{Peter-Jan H. S. Derks}\email{peter-janderks@hotmail.com}
\affiliation{Dahlem Center for Complex Quantum Systems, Freie Universit\"at Berlin, 14195 Berlin, Germany}

\author{Alex Townsend-Teague}
\affiliation{Dahlem Center for Complex Quantum Systems, Freie Universit\"at Berlin, 14195 Berlin, Germany}

\author{Ansgar G. Burchards}
\affiliation{Dahlem Center for Complex Quantum Systems, Freie Universit\"at Berlin, 14195 Berlin, Germany}

\author{Jens Eisert}
\affiliation{Dahlem Center for Complex Quantum Systems, Freie Universit\"at Berlin, 14195 Berlin, Germany}

\affiliation{Helmholtz-Zentrum Berlin f{\"u}r Materialien und Energie, 14109 Berlin, Germany}

\date{June 2024}

\maketitle

\begin{abstract}
Quantum error-correcting codes, such as subspace, subsystem, and Floquet codes, are typically constructed within the stabilizer formalism, which does not fully capture the idea of fault-tolerance needed for practical quantum computing applications.
In this 
work, we explore the remarkably powerful formalism
of detector error models, which fully captures fault-tolerance at the circuit level.
We introduce the detector error model formalism in a pedagogical manner and provide several examples.
Additionally, we apply the formalism to three different levels of abstraction in the engineering cycle of fault-tolerant circuit designs: finding robust syndrome extraction circuits, identifying efficient measurement schedules, and constructing fault-tolerant procedures.
We enhance the surface code's resistance to measurement errors, devise short measurement schedules for color codes, and implement a more efficient fault-tolerant method for measuring logical operators.
\end{abstract}

\section{Introduction}

\begin{figure}
\includegraphics[width=\columnwidth]{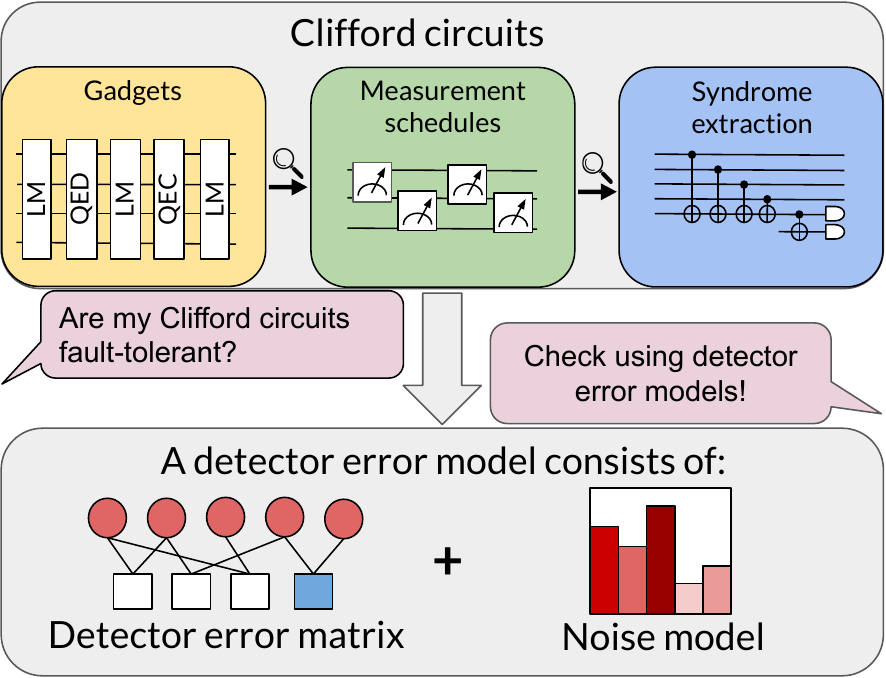}
\caption{Fault-tolerant quantum computation can be performed using fault-tolerant gadgets, which contain measurement schedules, which contain syndrome extraction circuits.
A detector error model contains all information necessary to verify the error-correcting capabilities of Clifford circuits subject to a noise model at these three levels.}
\label{fig:dem_figure}
\end{figure}

In general any quantum computation is subject to noise originating from uncontrollable interactions of the computer with its environment.
In order to perform large-scale quantum computations of arbitrary length and accuracy, resilience to this noise is ultimately required. 
This can be achieved by performing operations on qubits encoded in a \emph{quantum error correction} (QEC) code \cite{roffe2019quantum,RevModPhys.87.307,Roads,Roadmap}. 
On the circuit level, measurements which detect errors in QEC codes can be implemented by entangling additional auxiliary qubits with the data qubits and then measuring the auxiliary qubits. 
Circuits performing this task are referred to as \textit{syndrome extraction circuits} and typically make use of one and two-qubit operations. 
Syndrome extraction circuits must not necessarily use only a single auxiliary qubit and output a single measurement result. 
In fact, usage of multiple auxiliary qubits can improve error-correcting capabilities or relax hardware connectivity requirements \cite{chamberland2020topological}.

However, a combination of a QEC code and syndrome extraction circuit does not yield a fully fault-tolerant quantum computing scheme \cite{gottesman1997stabilizer,Roads}, or even an error-corrected quantum memory \cite{RevModPhys.87.307}.
One way to implement fault-tolerant QEC, is to combine a code with a \textit{measurement schedule}: a schedule that determines which sets of operators to measure in which order.
A measurement schedule can be represented
as a circuit, in which multi-qubit measurement
gates represent syndrome extraction circuits.

In order to go beyond quantum memories and also enable computation on encoded qubits, procedures generating \textit{fault-tolerant gadgets} are required \cite{gottesman1997stabilizer}. 
Gadgets perform operations such as measurement, state preparation, and gate implementation, and satisfy certain rules regarding how errors may spread within the gadget. 
Procedures generating gadgets for arbitrary QEC codes can be represented as circuits.
The structure of the circuit does not depend on the specific QEC code being used, as the details of the syndrome extraction circuits and measurement schedule are abstracted away.
Optimizing these procedures can lead to faster computation with substantially lower overhead.

The noise model one has to consider can differ depending on whether one is designing syndrome extraction circuits, measurement schedules, or fault-tolerant gadgets. 
For syndrome extraction circuits, a circuit-level noise model is required. 
For measurement schedules, if measurements are performed using fault-tolerant syndrome extraction circuits, one only needs to consider a \textit{phenomenological noise model} - one consisting of errors on data qubits in between syndrome extraction circuits and measurement errors.
For fault-tolerant gadgets, even cruder noise models can be used that distinguish errors only by when but not on which qubit they occur.

Circuits at all three levels of abstraction employ Pauli measurements to detect errors. 
A decoder has the task of interpreting the measurement outcomes in order to infer which errors have occurred.
A surprisingly powerful framework for passing this information to a decoder is that of \emph{detector error models} \cite{gidney2021stim, higgott2023sparse}.
By looking at the detector error model one can visualize the error-correcting abilities of a circuit.
This work is aimed at demonstrating the usefulness of detector error models, and their role in devising fault-tolerant quantum computing schemes. 
Specifically, we will show that detector error models can be usefully employed at all three previously mentioned levels of abstraction. 
An illustration of this main idea is shown in Fig.~\ref{fig:dem_figure}.

The remaining part of this work is structured as follows:
Section \ref{sec:intro_detector_error_models} contains an introduction to detector error models where the presentation is kept deliberately simple, providing a pedagogical introduction to the subject. 
In Section \ref{sec:syndrome_extraction_circuits}, we show how detector error models can be used to design syndrome extraction circuits for a device with high measurement noise bias. 
Section \ref{sec:measurement_schedule} employs the same framework to design measurement schedules for color codes while in Section \ref{sec:procedure}, we introduce a more efficient fault-tolerant logical measurement gadget. 
In Section \ref{sec:conclusion} we present an overview of our findings and propose directions for future work.

\section{Representing noisy Clifford circuits as detector error models}
\label{sec:intro_detector_error_models}

A common and useful class of quantum circuits is that of Clifford circuits.
Such circuits consist of Clifford quantum gates, Pauli measurements, classically controlled Pauli gates as well as state preparation in the computational basis. 
In particular, Clifford circuits can be used to implement arbitrary stabilizer quantum error correction codes, meaning that they can prepare codewords as well as perform the required measurements of stabilizers and logical $X, Y$ and $Z$ operators.

\subsection{Detectors}
\label{subsection:detectors}
Clifford circuits can contain measurements of stabilizers, which are Pauli operators of which the qubits are in a $+1$ or $-1$ eigenstate.
Their measurement outcomes are deterministic and can be used as what are called \textit{detectors} in this context.

\begin{definition}[Detector]
A \textit{detector} is a parity constraint on a set of measurement outcomes.
\end{definition}

\begin{example}[Finding detectors] 
$\triangleright$\footnote{We use $\triangleright$ to denote the start of an example, and $\triangleleft$ to denote the end.}
Deterministic measurements can be found by tracking stabilizers through a circuit. 
For a beginner-friendly introduction to stabilizer tracking, see Ref.~\cite{gottesman1998heisenberg}. 
For a description of more resource efficient methods, see Refs.~\cite{PhysRevA.70.052328,PhysRevA.73.022334,gidney2021stim}.
Consider now the circuit 
\begin{equation}
\includegraphics[valign=c]{three\_measurement\_circuit.pdf}.
\end{equation}
The single-qubit measurements in this circuit, as well as throughout the rest of this section, are performed in the Pauli $Z$-basis.
Before any gates in the circuit are applied, the stabilizer group of the input state is given by
\begin{equation}
\langle Z_1, X_2, -X_3 \rangle.
\end{equation}
After application of the first CNOT gate this becomes
\begin{equation}
\langle Z_1 Z_2 , X_1X_2, -X_3\rangle.
\end{equation}
Because the first measurement is $Z_1$, which is not contained in the stabilizer group, the measurement outcome will be non-deterministic.
We denote the measurement result with $m_1$ which assumes the value 0 (1) if the first qubit is measured to be in the +1 (-1) eigenstate of $Z_{1}$.
Following the update rules for the stabilizer group \cite{gottesman1998heisenberg} three updates are made to the stabilizer group; $X_1X_2$ is replaced by $(-1)^{m_1} Z_1$, $Z_1Z_2$ is multiplied by $(-1)^{m_1} Z_1$ and we stop tracking the stabilizers on the first qubit.
The resulting stabilizer group is given by
\begin{equation}
\langle(-1)^{m_1} Z_2, -X_3\rangle.
\end{equation}

The stabilizer group after application of the second CNOT gate is then
\begin{equation}
\langle -X_2 X_3, (-1)^{m_1} Z_2 Z_3\rangle.
\end{equation}
As in the first case, the Pauli operator being measured, $Z_2$, is not contained in the stabilizer group and yields a non-deterministic measurement outcome.
Given the measurement outcome $m_2$, the stabilizer group immediately before the last measurement is 
\begin{equation}
\langle (-1)^{m_1 \oplus m_2} Z_3 \rangle,
\end{equation}
where $\oplus$ denotes binary addition. When $Z_3$ is measured the outcome will thus be deterministic and equal to $m_1 \oplus m_2$. 
We conclude that this circuit has a single detector $d_1$ which represents the constraint that the parity of the first three measurements is even. 
We either denote a detector $d_1$ as a binary sum representing its parity constraint,
as in $d_1: m_1 \oplus m_2 \oplus m_3 = 0$,
or write it in the form of a binary vector
\begin{equation} \mathbf{d}_1 = \begin{bmatrix}
  1 \\
  1 \\
  1   
\end{bmatrix}.
\end{equation}
We denote the measurement outcomes of a circuit as a binary vector
\begin{equation} \mathbf{m} = \begin{bmatrix}
           m_{1} \\
           m_{2} \\
           m_{3}
\end{bmatrix}.
\end{equation}  
\triangleqed
\end{example}
In the example above, we have tracked the stabilizer group to derive detectors, alternatively detectors can be derived graphically using the ZX-calculus \cite{delafuente2024xyzrubycodemaking, McEwen2023relaxinghardware, Bombin_2024, rodatz2024floquetifyingstabilisercodesdistancepreserving}. We say that a detector $\mathbf{d}_i$ is violated by a set of measurement outcomes $\mathbf{m}$ if 
\begin{equation}
\mathbf{d}_i^T \mathbf{m} \neq b_i,
\end{equation}
where $b_i$ denotes the deterministic value assumed during noise-free implementation of the circuit.
Throughout, we assume \(b_i = 0\) for all detectors, without loss of generality.  
To justify this, let $\mathbf{m}^*$ be a fixed vector satisfying  
$\mathbf{d}_i^T \mathbf{m}^* = b_i \quad \forall i.$ Define the shifted vector
$
\tilde{\mathbf{m}} = \mathbf{m} + \mathbf{m}^*.
$  
Then  
$$
\mathbf{d}_i^T \tilde{\mathbf{m}} = \mathbf{d}_i^T \mathbf{m} + b_i,
$$
so the condition $\mathbf{d}_i^T \mathbf{m} = b_i\ \forall i$ is equivalent to $\mathbf{d}_i^T \tilde{\mathbf{m}} = 0\ \forall i$.  
Thus one can assume $b_i = 0$ without loss of generality by replacing $\mathbf{m}$ with $\tilde{\mathbf{m}}$.

\subsection{The detector matrix}
Any given circuit may contain multiple deterministic measurements. 
In general, it is desirable to construct the maximum number of linearly independent detectors in order to improve the fault-tolerant properties of a circuit.
From this point onward, we assume that the maximum number of detectors is used.
Each deterministic measurement will give rise to one linearly independent detector.
A useful representation of the detectors of a circuit is as a \textit{detector matrix}.

\begin{definition}[Detector matrix]
\label{detectormatrix}
A \textit{detector matrix} $D$ is a binary matrix in $\zz_2^{d\times m}$ whose rows are the linearly independent detectors of a circuit.
The number of rows $d$ is equal to the number of deterministic measurements.
The number of columns $m$ is equal to the total number of measurements.
\end{definition}

Note that the kernel of a detector matrix is all possible measurement outcomes in the absence of noise.
Detector matrices have previously been defined as the outcome code of a circuit in  Ref.~\cite{delfosse2023spacetime}. 

\begin{example}[Constructing a detector matrix] \label{rep_code_example} $\triangleright$

\begin{equation}
\includegraphics[valign=c]{repcode\_noiseless\_small.pdf}.
\end{equation}

We will construct the detector matrix for the above circuit and see that it can be used to identify if and where possible errors occurred during the circuit implementation. 
Initially, the three central qubits are stabilized by (among other terms) $Z_2Z_3$ and $Z_3Z_4$.
The first two measurements $M_1, M_2$ measure operators $Z_2Z_3$ and $Z_3Z_4$, a fact that can be expressed via the circuit equality
\begin{equation}
\includegraphics[valign=c]{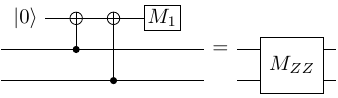}
\end{equation}
where the circuit on the left side of the equation is commonly referred to as a syndrome extraction circuit \cite{roffe2019quantum}.
Because $M_1$ and $M_2$ are measuring stabilizer elements, their outcomes are deterministic and can be used to create the detectors
\begin{align}
d_1\colon m_1 = 0 ,\qquad d_2\colon m_2 = 0 .
\end{align}
Though the following measurement $M_3$ is not deterministic, the subsequent two are, giving rise to a parity constraint. We may therefore 
construct the additional detectors 
\begin{align}
d_3\colon m_3 \oplus m_4 = 0, \qquad d_4\colon m_4 \oplus m_5 = 0 .
\end{align}
All four detectors together are summarized in the form of the detector matrix
\begin{equation}
D_1 \colon=
\begin{bmatrix}
1 & 0 & 0 & 0 & 0 \\
0 & 1 & 0 & 0 & 0 \\
0 & 0 & 1 & 1 & 0 \\
0 & 0 & 0 & 1 & 1 \\
\end{bmatrix}
\label{eq:detector_matrix_1}.
\end{equation}
To see the utility of this matrix, consider the occurrence of an $X$ error just before measurement $M_3$
\begin{equation}
\includegraphics[valign=c]{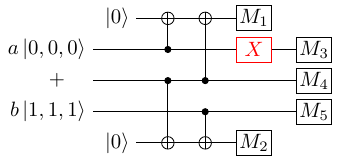}.
\end{equation}

Given no access to the other measurement outcomes, it is not possible to know if the $X$ error has occurred or not, because $M_3$ is non-deterministic in the absence of noise.
The possible measurement results of this circuit including the error are either of $[0, 0, 0, 1, 1]$ or $[0, 0, 1, 0, 0]$.
Therefore, the $X$ error violates $d_3$, in that we would now deterministically have $m_3 \oplus m_4 = 1$, rather than $m_3 \oplus m_4 = 0$ as we would expect in the absence of noise.
This is independent of whether the outcome $m_3$ is 0 or 1.

The attentive reader may already have noticed that there is some freedom in specifying a set of detectors. 
For example, above one could have chosen equivalent third and fourth detectors which depend additionally on the first and second measurement outcomes
\begin{align}
d_{3}': m_1\oplus m_3 \oplus m_4 = 0,\nonumber \\
\quad d_{4}':m_2 \oplus m_4 \oplus m_5=0. 
\end{align}
This choice gives rise to the following alternative detector matrix 
\begin{equation}
D_2 :=
\begin{bmatrix}
1 & 0 & 0 & 0 & 0 \\
0 & 1 & 0 & 0 & 0 \\
1 & 0 & 1 & 1 & 0 \\
0 & 1 & 0 & 1 & 1 \\
\end{bmatrix}
\label{detector_matrix_2}.
\end{equation}
We explain the consequences and rules for choosing different detector matrices in Subsection \ref{subsec:detector_error_models}.
\triangleqed
\end{example} 

One can construct a detector matrix by performing a stabilizer tableau simulation. 
This can be performed in $O(ng + nd + n^2r)$ time, where $n$ is the number of qubits, $g$ is the number of gates, $d$ is the number of deterministic measurements and $r$ is the number of non-deterministic measurements \cite{gidney2021stim}. 
We simplify this here to worst-case complexity $O(n^3l)$, where $l$ is the number of circuit layers.

\subsection{Error vectors}

In the previous subsection, we introduced detector matrices and gave an example of their use in the case of a single circuit error occurring. 
In this subsection, we formalize errors in the detector error model framework and introduce the vector representation of a noise model.

\begin{definition}[Errors]
An error $E_i$ refers either to an undesired Pauli operator acting nontrivially on one or more qubits at a specific circuit location, or to a flip in the result of a specific measurement.
\end{definition}

A set of stochastic errors forms a noise model:

\begin{definition}[Noise model]
We represent a \textit{noise model} as a column vector $\mathbf{p}_e \in [0,1]^{e}$
where entry $p_i$ denotes the probability of error $E_i$ occurring and $e$ is the total number of errors that may occur.
\end{definition}

We assume independence between all errors, i.e.\ for any $i \neq j$, 
\[
Pr(E_i \wedge E_j) = p_ip_j.
\]

\begin{example}[A circuit's noise model] $\triangleright$
Consider again the circuit introduced in the previous example, now with 8 possible $X$-error locations
\begin{equation}
\includegraphics[valign=c]{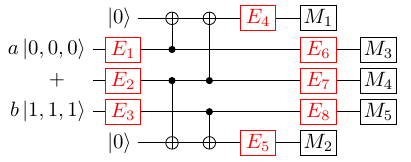}  .  
\end{equation}
This noise model is represented by a column vector with 8 entries whose values are $p_1$ to $p_8$. 
If we switch to a more complicated noise model in which an $X,Y$ or $Z$ error could have occurred at each of the 8 circuit locations, the column vector would require 24 entries instead
\label{example:circuit_error_model}
\triangleqed
\end{example} 

In the preceding examples, we examined a noise model where $X$ Pauli errors follow certain gates. 
In general, errors can occur after any gate. 
Thus, demonstrating a circuit's fault tolerance requires showing its ability to correct any error, regardless of where it occurs. 
This type of noise model is known as a \textit{circuit-level noise} model.

\begin{definition}[Circuit-level noise]
In a \textit{circuit-level noise} model any $n$-qubit Pauli error can occur after a $n$-qubit gate and any single-qubit Pauli error after state preparation.
A classical bit-flip error can occur on measurement outcomes \footnote{If the measurement is non-destructive, a Pauli error can occur on the measured qubit.}. 
\label{def:circuit_level_noise}
\end{definition}

A circuit-level error is any error that can occur in a circuit-level noise model.
A noisy Clifford circuit can be simulated by drawing samples from an error model.
We represent such a sample as a \textit{circuit error vector}. 

\begin{definition}[Circuit error vector]
A circuit error vector is a length $e$ column vector $\mathbf{e}$, whose $i$-th entry is $1$ if error $E_i$ occurred, and $0$ otherwise.
\end{definition}

\subsection{Measurement syndrome matrix}

To calculate which detectors are violated by a circuit error, one first needs to calculate which measurements are affected. 
This can be done by multiplying the \textit{measurement syndrome matrix} with the circuit error vector. 

\begin{definition}[Measurement syndrome matrix]
A \textit{measurement syndrome matrix} is a binary matrix $\Omega\in \zz_2^{m\times e}$ where $m$ is the number of measurements and $e$ is the number of errors that may occur in a circuit.
The entry $\Omega_{i,j}$ is $1$ if error $j$ flips $m_i$ and $0$ otherwise.
\end{definition}

An error $E_j$ flips a measurement outcome, if the resulting operator found by propagating the error through the circuit anti-commutes with the Pauli operator being measured.
A measurement syndrome matrix can be calculated by propagating errors through the circuit using Pauli frames \cite{rall2019simulation}.
For a circuit with $l$ gate layers acting on $n$ qubits, updating a Pauli frame throughout the circuit takes $O(nl)$ time  \cite{gidney2021stim}.
If there is a noise location after each gate, $O(nl)$ Pauli frames need to be simulated.
Therefore,  the time complexity of constructing a measurement syndrome matrix is $O(n^2l^2)$ \cite{gidney2021stim}.

\begin{example}[Constructing a measurement syndrome matrix] 
$\triangleright$
Let us again consider the circuit in Example \ref{example:circuit_error_model} in which an $X$ error can occur at any of the 8 highlighted error locations. The measurement syndrome matrix 
is
\begin{equation}\label{eq:measurement_syndrome_matrix}
\Omega_1 \coloneqq
\begin{bmatrix}
1 & 1 & 0 & 1 & 0 & 0 & 0 & 0\\
0 & 1 & 1 & 0 & 1 & 0 & 0 & 0\\
1 & 0 & 0 & 0 & 0 & 1 & 0 & 0\\
0 & 1 & 0 & 0 & 0 & 0 & 1 & 0\\
0 & 0 & 1 & 0 & 0 & 0 & 0 & 1\\
\end{bmatrix}.
\end{equation}

To gain some intuition, consider entry $\Omega_{1,1}$: This entry is 1 because $E_1$ propagates to a weight-two $X$ operator acting on the top two qubits.
This flips $m_1$.
\triangleqed
\end{example}

We can now state the violation condition of a detector in terms of the syndrome matrix.

\begin{definition}[Violating detectors]
A circuit error $\mathbf{e}$ \textit{violates} a detector $\mathbf{d}$ if $\mathbf{d}^T \Omega \mathbf{e} = 1$.
\end{definition}

\subsection{Detector error models}\label{subsec:detector_error_models}
We now combine the concepts of detectors and errors introduced in the previous subsections. 
A \textit{detector error model} consists of two parts: a \textit{detector error matrix} containing full information about which errors violate which detectors, and a \textit{noise model} assigning a likelihood to the occurrence of each error. 

\begin{definition}[Detector error matrix]
A \textit{detector error matrix} is a binary matrix $H\in \zz_2^{d \times e}$, where $d$ is the number of detectors and $e$ is the number of errors that can occur in a noise model. The entry $H_{i,j}$ equals $1$ if detector $i$ is violated by error $j$  and $0$ otherwise.
\end{definition}

Detector error matrices can be obtained as products of detector matrices and measurement syndrome matrices via $H = D \Omega$.

\begin{example}[Constructing a detector error matrix] $\triangleright$
We can construct the detector error matrix of the circuit in Example \ref{example:circuit_error_model} by multiplying the detector matrix defined in Eq.~\eqref{eq:detector_matrix_1}
and the syndrome measurement matrix in Eq.~\eqref{eq:measurement_syndrome_matrix} as
\begin{align}
H_1 &\coloneqq D_1 \Omega_1 \nonumber \\
&=
\begin{bmatrix}
1 & 1 & 0 & 1 & 0 & 0 & 0 & 0\\
0 & 1 & 1 & 0 & 1 & 0 & 0 & 0\\
1 & 1 & 0 & 0 & 0 & 1 & 1 & 0\\
0 & 1 & 1 & 0 & 0 & 0 & 1 & 1\\
\end{bmatrix}.
\label{eq:detector_error_matrix_example}
\end{align}
\triangleqed
\end{example}
A detector error matrix can be represented as
a \emph{Tanner graph} \cite{tanner1056404, loeliger2004introduction}.  
A Tanner graph is a useful graph representation with two types of nodes:
\textit{bit nodes} that represent errors and \textit{check nodes} that represent detectors. 
A bit node $i$ is connected via an edge to a check node $j$ if entry $(i,j)$ in the detector matrix is $1$.
The detector error matrix in Eq.~\eqref{eq:detector_error_matrix_example} can be represented by the  Tanner graph
\begin{equation}
\includegraphics[width=0.85\columnwidth, valign=c]{rep\_code\_graph.pdf}
\label{eq:dense_tanner_graph}.
\end{equation}

We refer to the set of detectors that are violated by a circuit error $\mathbf{e}$ as a \textit{syndrome}, which is given by $\mathbf{s}
=H\mathbf{e}$.
The task of a \textit{decoder} is to infer which circuit error vector $\mathbf{e}$ occurred, given a syndrome $\mathbf{s}$.
From the Tanner graph above, one can see that every individual error violates a unique set of detectors.
Therefore, any single error can be correctly inferred from its syndrome and can be corrected.

In Subsection \ref{subsection:detectors}, we have shown an alternative choice of detectors leading to detector matrix $D_2$ in Eq.~\eqref{detector_matrix_2}.
Using this alternative detector matrix we find the different detector error matrix
\begin{align}
H_2 &= D_2 \Omega_1 \nonumber \\
&= \begin{bmatrix}
1 & 1 & 0 & 1 & 0 & 0 & 0 & 0\\
0 & 1 & 1 & 0 & 1 & 0 & 0 & 0\\
0 & 0 & 0 & 1 & 0 & 1 & 1 & 0\\
0 & 0 & 0 & 0 & 1 & 0 & 1 & 1\\
\end{bmatrix}.
\end{align}
\label{eq:detector_error_matrix_example_2}
This detector error matrix can be represented by the following sparser Tanner graph.

\begin{equation}
\includegraphics[width=0.85\columnwidth, valign=c]{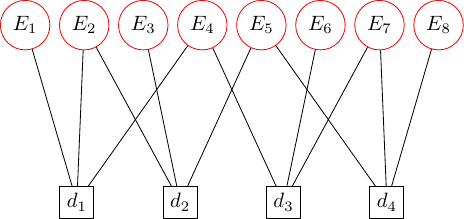}
\label{eq:sparse_tanner_graph}.
\end{equation}

Again this Tanner graph makes apparent the fact that every individual error violates a unique set of detectors.
Is this by coincidence or does any choice of detectors lead to a detector error matrix that can uniquely identify each error?
The latter is the case.
In fact, the error-correcting capabilities of a circuit are independent of which detector matrix is used:

\begin{lemma}[Freedom in choosing detector matrices]
\label{thm:freedom_in_choosing_detectors}
Two detector matrices of the same circuit can distinguish between the same pairs of error sets.
\end{lemma}
\begin{proof}
The kernel of a detector error matrix, $\text{ker}(H) = \text{ker}(D\Omega)$, is formed by all possible measurement outcomes of a noiseless circuit.
Therefore two detector error matrices of the same circuit have the same kernel.
A detector matrix $D_1$ can distinguish between circuit error vectors $\mathbf{e}_1, \mathbf{e}_2$ if and only if $\mathbf{e}_1 - \mathbf{e}_2 \notin \text{ker}(D_1 \Omega).$ 
As $\mathbf{e} \in \text{ker}(D_1 \Omega)$ is equivalent to $\Omega \mathbf{e} \in \text{ker}(D_1)$ any other detector matrix $D_2$ with identical kernel can distinguish the same pairs of circuit error vectors.

\end{proof}

We will comment on consequences  of different detector matrices at the end of this section.

\subsection{Observables}

To obtain the outcomes of encoded computations, logical operators need to be measured.
In some approaches to encoded computation, logical operators are also measured non-destructively at intermediate stages of the circuit.
Often logical operators are not measured directly, instead multiple single-qubit measurements are performed from which the value of a logical operator can be inferred.
We refer to a set of measurements from which one can infer the value of a logical operator as an \textit{observable}. More formally:
\begin{definition}[Observable] 
An observable is a binary sum of measurements, which equals the outcome of measuring a logical operator, for any logical state.
\end{definition}

Similar to detectors, an observable $o_i$ can be written as a column vector $\mathbf{o}_i$.
If an observable has a deterministic value, we again denote it with $b_i$.
An error $\mathbf{e}$ violates an observable if $\mathbf{o}_i^T \Omega \mathbf{e} \neq b_i$.
The fidelity of logical gates can be found by running \textit{benchmarking circuits}.
For a benchmarking circuit, the value of an observable is known in a noise-free implementation.
To find the logical gate fidelity, one performs the noisy benchmarking circuit many times.
If the circuit is run on quantum hardware, the only outcome of each run is a vector of measurement outcomes $\mathbf{m}$.
From these measurement outcomes, the syndrome $\mathbf{s} = D \mathbf{m}$ is calculated and passed to the decoder.
The decoder infers which error $\mathbf{e}_\text{inferred}$ has occurred given the syndrome and the detector error matrix, $\mathbf{s} = H \mathbf{e}_\text{inferred}$.
From this it can be predicted if an observable is violated: $\mathbf{o}_i^T \Omega \mathbf{e}_\text{inferred}$.
A \textit{logical error} has occurred if \begin{equation}\mathbf{o}^T_i \mathbf{m} \oplus \mathbf{o}_i^T \Omega \mathbf{e}_\text{inferred} \neq b_i.
\end{equation}
Note that when running a simulation on a classical computer, one can know specifically that error $\mathbf{e}_\text{sample}$ occurred. In this case one could equivalently write that a logical error occurs if
\begin{equation}\mathbf{o}^T_i \Omega \mathbf{e}_\text{sample} \oplus \mathbf{o}_i^T \Omega \mathbf{e}_\text{inferred} \neq b_i.
\end{equation}
The \textit{logical error rate} is defined as the number of logical errors divided by the number of runs.

\begin{example}[Repetition code memory experiment] $\triangleright$
Consider a noise model that inserts single-qubit $X$ errors into a circuit with probability $p$.
The simplest code that can protect quantum information against this type of noise is the $d=3$ repetition code with logical basis state vectors $\ket{0}_{L} = \ket{0,0,0}$ and $\ket{1}_{L} = \ket{1,1,1}$.
A logical $Z$ basis measurement can be performed by simply measuring any of the three qubits in the $Z$ basis.
In the following circuit, the qubits start in either the state vector $\ket{0}_{L}$ or $\ket{1}_{L}$.
\begin{equation}
\includegraphics[width=0.85\columnwidth, valign=c]{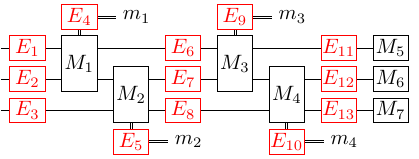}.
\end{equation}

In the circuit, 
the first four gates labeled $M_1, M_2, M_3$ and $M_4$ represent $ZZ$ measurements. 
Errors $E_4, E_5, E_9$ and $E_{10}$ are classical bit flip errors that may flip measurement outcomes $m_1, m_2, m_3$ and $m_4$.
A choice of detectors for this circuit is
\begin{align}
d_1&\colon m_1 = 0, \nonumber \\
d_2&\colon m_2 = 0, \nonumber \\
d_3&\colon m_1 \oplus m_3 = 0 ,\nonumber \\
d_4&\colon m_2 \oplus m_4 = 0 ,\label{eq:detectors_rep_code}
\\ 
d_5&\colon m_3 \oplus m_5 \oplus m_6 = 0, \nonumber \\
d_6&\colon m_4 \oplus m_6 \oplus m_7 = 0 . \nonumber
\end{align}
This circuit has one observable.
Two possible sets of measurements, among many choices, that can be used as observable here are $\{m_5\}$ or $\{m_5, m_6, m_7\}$.
The reason that these sets of measurements can be used is that their parity depends on whether the qubits start in $\ket{0}_{L}$ or $\ket{1}_{L}$.
Note that the parities of the detectors in Eq.~\eqref{eq:detectors_rep_code} do not depend on whether the qubits start in $\ket{0}_{L}$ or $\ket{1}_{L}$.
If we define the qubits to start 
in $\ket{0}_{L}$, the observable's parity is fixed; $o_1 \colon m_5=0$ or $o'_1 \colon m_5 \oplus m_6 \oplus m_7 = 0$ .
Together, the 
detectors and observables form a maximal set of linearly independent constraints on measurement outcomes.
Note that any set of detectors can be added to an observable to get another valid representation of the observable.
For example, here $o_1 + d_6 + d_4 + d_2 = o'_1$.

Below, in the Tanner graph of the circuit, we have added a blue check node to represent the observable $o_1 \colon m_5$.
One can verify that at least three errors are needed to flip the observable without violating any detectors. 
Examples of such errors are $\{E_1, E_2, E_3\}$ and $\{E_7, E_8, E_9, E_{11}\},$
\begin{equation}
\includegraphics[valign=c, width=0.8\columnwidth]{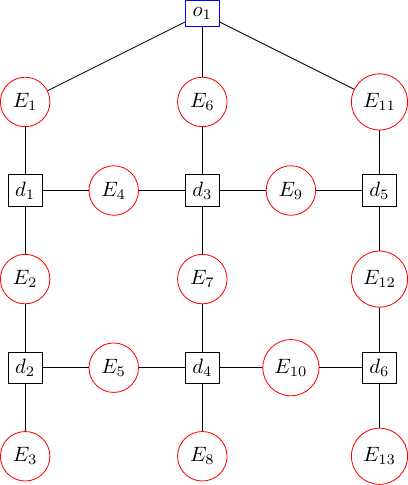}.
\end{equation}
\triangleqed
\end{example}

The circuit in the example above has \textit{circuit distance} 3.
\begin{definition}[Circuit distance]
The circuit distance is the minimum number of errors needed to flip at least one observable without violating any detectors.
\end{definition}

For a circuit that has circuit distance $d_\text{circ}$, any set consisting of $\lfloor(d_\text{circ} - 1) / 2\rfloor$ errors is \textit{correctable}.
Here, correctable means that there exists a decoder that can correctly predict whether any observable has been violated.
The circuit distance depends on the noise model the circuit is subject to.
Importantly, following directly from Lemma \ref{thm:freedom_in_choosing_detectors}, it does not depend on which detectors are chosen.

\begin{corollary}[Preserved circuit distance]
The circuit distance is unaffected by the choice of detector matrix.
\end{corollary}

The choice of detector matrix primarily influences the performance of a decoder. 
Decoding is typically easier when using a Tanner graph with fewer edges. 
This is illustrated by the two Tanner graphs shown previously in Eq.~\eqref{eq:dense_tanner_graph} and Eq.~\eqref{eq:sparse_tanner_graph}, which demonstrate that even small circuits can have a significant difference in the number of edges. 
In general, reducing the number of edges in a Tanner graph can be done by applying elementary row operations that reduce the number of 1 entries on the associated detector error matrix.
It is not possible to find the detector error matrix with the least number of 1 entries efficiently.
This is an NP-hard problem as it is equivalent to finding the minimum weight basis of a vector space which is NP-complete \cite{gidney2022qec, deo1982algorithms}.
This is rather unfortunate, as a consequence of this is that, when using \textit{stim} \cite{gidney2021stim}, QEC researchers have to manually find detectors for each circuit they want to simulate \cite{gidney2022qec}.

\section{Designing syndrome extraction circuits for high measurement noise}
\label{sec:syndrome_extraction_circuits}

When implementing a QEC code on hardware, the goal is the best possible protection against the specific experimentally relevant noise model while simultaneously minimizing resource requirements.
Previous work has shown that QEC codes can be tailored to take advantage of the structure inherent in a Pauli noise model biased towards $X, Y$ or $Z$ noise \cite{dua2022clifford, tiurev2022correcting, roffe2022bias,newtiurev}.
The \emph{XZZX surface code} is an example of a code tailored to biased noise \cite{Bonilla-Ataides2020, darmawan2021practical}.
It is based on the \emph{CSS surface code}, a $[[d^2, 1, d]]$ topological QEC code \cite{kitaev2003fault} which is a promising candidate for a fault-tolerant quantum computing architecture due to its high threshold and the locality of its stabilizers \cite{dennis2002topological, fowler2012surface}.

To the best of our knowledge,
a type of noise bias that circuits implementing the surface code have not previously been tailored to is measurement noise bias.
Such noise is relevant in current quantum hardware -- for example, in superconducting qubit devices.
In fact, in a circuit-level noise model of Google's superconducting device, measurement errors are the most likely type of error to occur and five times more likely to occur than an error after a two-qubit gate  \cite{gidney2022benchmarking}.
Here, we focus on tailoring a circuit that implements the surface code towards high measurement bias.
The surface code itself cannot be tailored to this type of noise bias using conventional techniques such as applying Clifford conjugations to stabilizers.
This is because Clifford conjugations do not have an effect on which detectors are violated by measurement errors.

\begin{example}[Syndrome extraction circuits with two auxiliary
qubits]
$\triangleright$
To tailor the surface code to measurement noise we adapt the syndrome extraction circuit.
Our syndrome extraction circuit uses two auxiliary qubits that are measured with respect to the same basis. 
The syndrome extraction circuit measuring a weight-4 $X$ stabilizer is 
\begin{equation}
\includegraphics[valign=c, width=0.85\columnwidth]{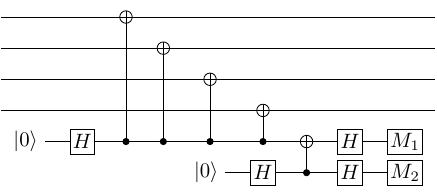}.
\end{equation}

Similarly, the circuit measuring a weight-4 $Z$ stabilizer is
\begin{equation}
\includegraphics[valign=c, width=0.85\columnwidth]{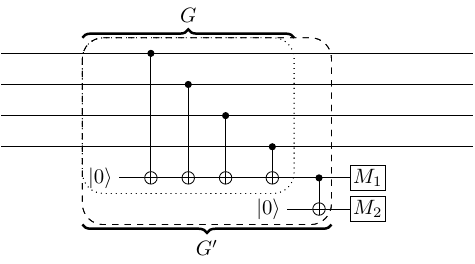} \label{eq:Z_stab}.
\end{equation}

The borders separating two sections of the circuit serve for later reference.
The circuits with two auxiliary qubits involve more gates, initializations, and measurements, but they use the same gate set as circuits that are conventionally used in experimental surface code implementations.
The reason why the circuits with two auxiliary qubits are good at detecting measurement errors can be understood by observing their detector error models. 
Specifically, we look at the detector error model of these circuits subject to pure measurement noise. Consider the following circuit which repeatedly measures $Z$-stabilizers

\begin{equation}
\includegraphics[valign=c, width=0.85\columnwidth]{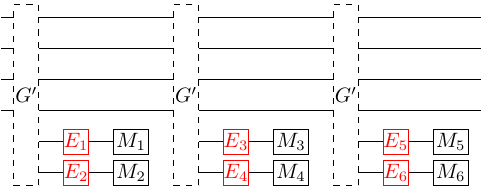}.
\end{equation}
The dashed-border gate represents the gates in the dashed box in Eq.~\eqref{eq:Z_stab} while red gates represent possible $X$-errors, i.e., errors that flip the measurement results. We define the detectors in this circuit as 
\begin{align}
\begin{alignedat}{2}
d_1\colon m_1 \oplus m_2 &= 0, \quad d_2\colon m_2 \oplus m_3 &&= 0,\\ 
d_3\colon m_3 \oplus m_4 &= 0, \quad d_4\colon m_4 \oplus m_5 &&= 0,\\ 
d_5\colon m_5 \oplus m_6 &= 0.
\end{alignedat}
\end{align}
This choice of detectors allows the use of a \emph{matching decoder} \cite{brown2021conservation}. 
The Tanner graph of the detector error model for this circuit is
\begin{center}
\begin{equation}
\includegraphics[valign=c]{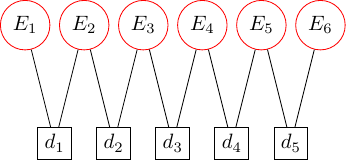}.
\end{equation}
\end{center}

Compare this to the standard circuit for repeatedly measuring $Z$-stabilizers given by
\begin{equation}
\includegraphics[valign=c, width=0.85\columnwidth]{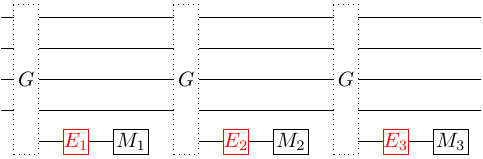}.
\end{equation}

The gate with a dotted border represents the gates in the dotted box in Eq.~\eqref{eq:Z_stab}. This circuit's detectors are 
\begin{equation}
d_1: m_1 \oplus m_2 = 0, \quad d_2: m_2 \oplus m_3 = 0 ,
\end{equation}
and its Tanner graph is given by
\begin{center}
\begin{equation}
\includegraphics[valign=c]{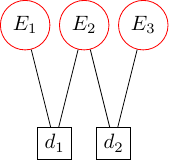}.
\end{equation}
\end{center}
As the respective Tanner graphs are those of classical repetition codes, it is apparent that six measurement errors are needed to avoid violating any detectors in the proposed circuit while in the standard circuit only three measurement errors are required.
\triangleqed
\end{example}

In the example, 
we have shown that the proposed syndrome extraction circuit is resilient to twice as many measurement errors as the standard circuit.
In Appendix \ref{sec:simulation_details}, 
we investigate the performance of the proposed syndrome extraction circuits in the presence of circuit-level noise. 
The source code is available at Ref.\ \cite{github_repo}.

We numerically find that for circuit-level noise modelling superconducting qubits the proposed circuits are approximately equally efficient in terms of qubit count (including auxiliary qubits) in the teraquop regime as the standard circuits.
The proposed circuits are more efficient for circuit-level noise with higher measurement bias.

\section{Designing color code measurement schedules}
\label{sec:measurement_schedule}

In the previous section, we focused on syndrome extraction circuits that measure Pauli operators. 
Now, our aim shifts towards a higher level of abstraction, concealing the details of the syndrome extraction circuits.
Instead, we focus on which stabilizers to measure at which time to perform QEC.
Here, we solely focus on performing QEC using CSS codes.
Before proceeding, we clarify what it actually means to perform QEC.

\subsection{Fault-tolerant QEC procedure}
The primary objective of an error correction procedure is the protection of logical information encoded into data qubits over an extended duration. 
Even when physical error probabilities are low, errors will accumulate over time and corrupt the encoded information. 
To counteract this accumulation, errors need to be intermittently removed by use of a \textit{QEC procedure}.
Such a procedure is described by a circuit that achieves the sought-after removal of errors that would have accumulated on the quantum state.
The main challenge in the design of QEC procedures is that the gates used to implement the procedure are themselves noisy, and thus  \textit{internal errors} can occur.
\begin{definition}[Internal, input and output errors]
\label{definition:internal_input_output_errors}
An error occuring before the start of the circuit implementing the procedure is classified as an input error.
Errors occuring during the circuit are termed internal errors.
Any error present at the end of the circuit is referred to as an output error.
\end{definition}

The use of a QEC procedure will only lead to a lower logical error probability if, in addition to correcting input errors, the spread of internal errors is sufficiently limited. 
A QEC procedure that achieves this is said to be fault-tolerant.

\begin{definition}[Fault-tolerant QEC procedure \cite{gottesman2010introduction}]
\label{definition:FTEC}
A QEC procedure is fault-tolerant to distance $d = 2t+1$ if, provided the sum of the number of input and internal errors is at most $t$, both the output error’s $X$ weight and the output error's $Z$ weight are at most the number of internal errors.
\label{definition:FTEC_protocol}
\end{definition}
Note that the output error's $X$ and $Z$ weight are defined to be the minimized $X$ and $Z$ weight over all stabilizer equivalent representatives.
Fig.~\ref{fig:qec_procedure} shows an illustration of a QEC procedure with 3 input errors and 4 internal errors. 
The $X$ weight of the output error is 2 and the $Z$ weight is 1.

\begin{figure}
\includegraphics[width=\columnwidth]{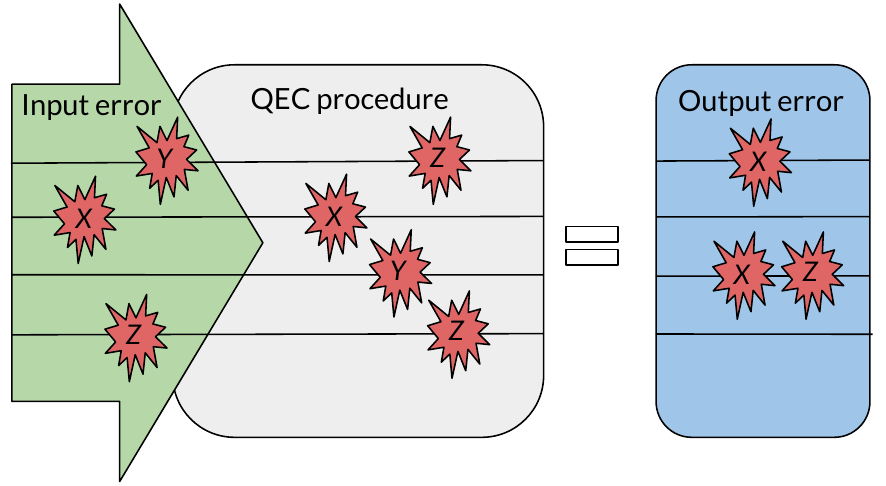}
\caption{Input errors occur before the QEC procedure and internal errors occur during the procedure. Together they lead to an output error.}
\label{fig:qec_procedure}
\end{figure}

There are many \textit{styles} of QEC procedures, with the most common ones being Steane-style \cite{steane1997active}, Knill-style \cite{knill2005scalable}, and Shor-style \cite{shor1996fault} QEC procedures.
In the next section, we focus on a variant of Shor-style, termed \textit{short Shor-style QEC procedures} \cite{delfosse2020short}.

\subsection{Fault-tolerant measurement schedules}

A short Shor-style QEC procedure can be implemented using a \emph{measurement schedule} combined with a decoder. 

\begin{definition}[Measurement schedule]
A measurement schedule for an $[n, k, d]$ code, denoted $\mathcal{C}$, consists of $n_M$ rounds of measurements.
In round $i$ of the schedule, a set of Pauli operators acting on disjoint sets of data qubits of $\mathcal{C}$ is measured, resulting in a vector $\mathbf{r}_i$ of measurement outcomes, where $i \in [1, 2, \dots , n_M]$.
\end{definition}

\begin{example}[Measurement schedule for the three-qubit repetition code]
$\triangleright$
In the remainder of this section, we will use the three-qubit repetition code as an example.
Below is a three-round measurement schedule, with a single measurement performed in each round
\begin{equation}
\includegraphics[valign=c, width=0.85\columnwidth]{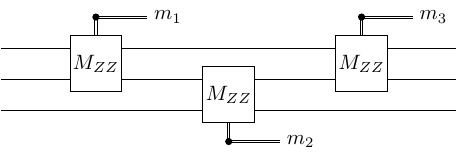}
\label{eq:rep_code_measurement_schedule_noiseless}.
\end{equation} 
\triangleqed
\end{example}

For CSS codes, the measurement schedule can be split into two parts: 
One for correcting $X$ errors and the other for $Z$ errors.
Using the definition of a fault-tolerant QEC procedure, we are now going to define a fault-tolerant measurement schedule for correcting $X$ errors.
To do this we first need to introduce a noise model.

\begin{definition}[Phenomenological noise model]
The phenomenological noise model considered here consists of two types of errors: measurement errors and $X$ errors on data qubits.
The $X$ errors can occur on data qubits before each measurement round. 
We refer to the $X$ errors before the first measurement round as input errors, and any error (including measurement errors) after the first measurement round as internal errors.
\end{definition}

With this noise model it may seem that we are cutting corners and oversimplifying the task of designing fault-tolerant circuits.
This, however, is not true.
It turns out that using a measurement schedule that satisfies Definition \ref{definition:FTEC} for the phenomenological noise model, we can create a circuit that automatically also satisfies Definition \ref{definition:FTEC} for a circuit-level noise model.
The QEC procedure that creates the circuit-level noise resilient circuit is summarized in Fig.~\ref{fig:qec_procedure_components}.
The procedure consists of two main steps, namely creating a $Z$ measurement schedule from the $X$ measurement schedule and implementing the measurements using \textit{fault-tolerant syndrome extraction circuits}.

Here, we focus on designing the measurement schedule, as this is the only code dependent part of the procedure.
For a detailed explanation of the QEC procedure and a proof that the resulting circuit is fault-tolerant for circuit-level noise, see Appendix \ref{sec:justification_noise_model}. While it is useful to consider a phenomenological noise model when designing a measurement schedule, it is important to use a circuit-level noise model for simulation. Simulations of phenomenological noise are not a good indicator of performance under circuit-level noise.
 
\begin{figure*}
\centering
\includegraphics[width=\linewidth]{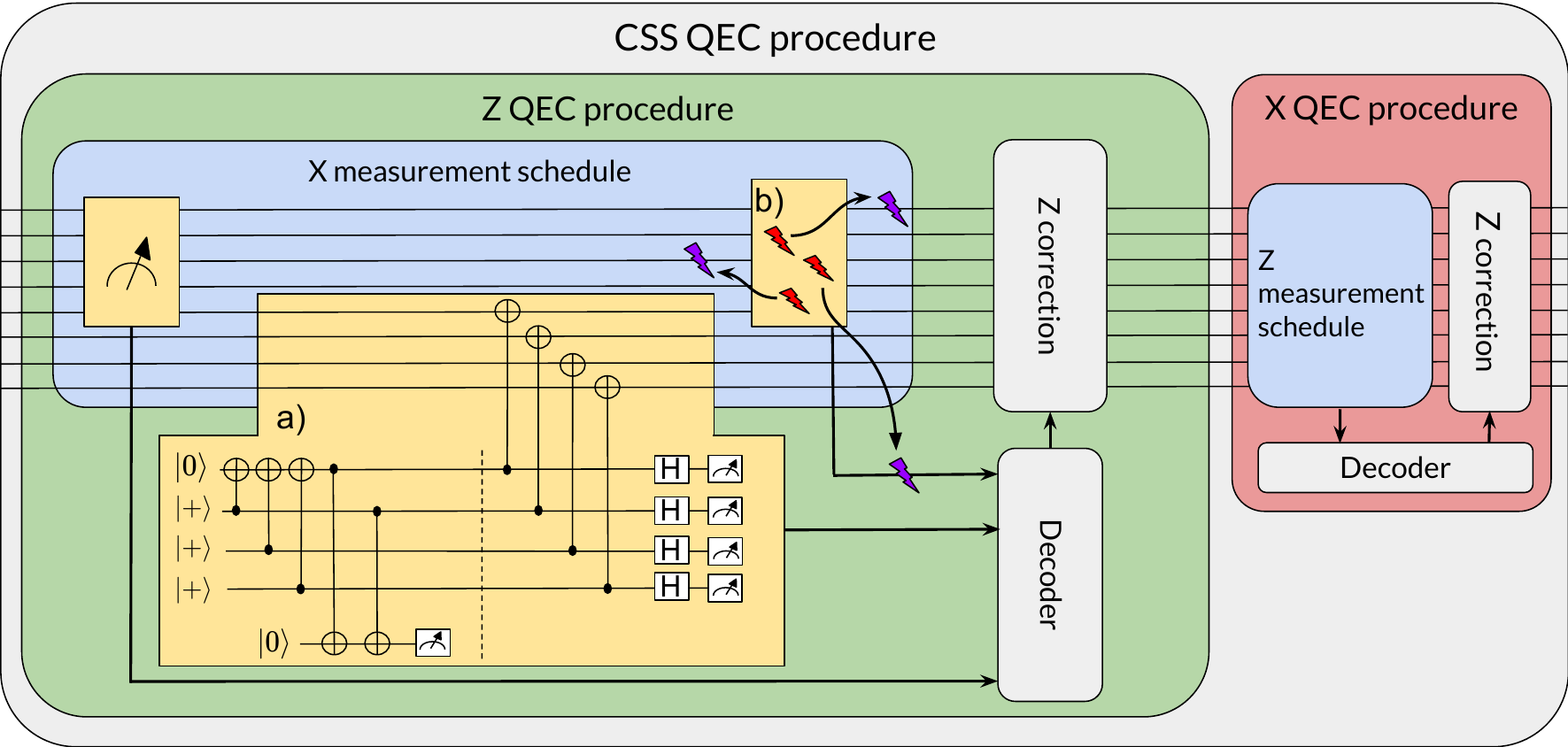}
\caption{The CSS QEC procedure is split up into two parts, one for correcting $X$ errors and one for correcting $Z$ errors. 
Inside each of these parts is a measurement schedule, whose measurement outcomes are passed to a decoder, which outputs a correction.
a) Measurements are performed using fault-tolerant syndrome extraction circuits. 
b) Errors happening during the syndrome extraction circuit are equivalent to errors that can happen in the phenomenological noise model.}
\label{fig:qec_procedure_components}
\end{figure*}

\begin{example}[Phenomenological noise model for the three-qubit repetition code]
$\triangleright$
The phenomenological noise model for the three-round measurement schedule is illustrated below
\begin{equation}
\includegraphics[valign=c, width=0.85\columnwidth]{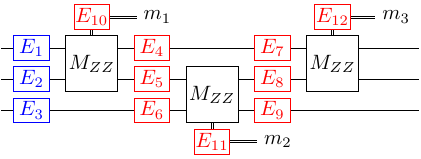}.
\label{eq:rep_code_measurement_schedule_initial}
\end{equation}
The blue boxes represent input error locations and the red boxes represent internal errors.
We simplify the noise model by merging equivalent error locations.
Merging means treating multiple error locations as a single error location.
The result is:
\begin{equation}
\includegraphics[valign=c, width=0.85\columnwidth]{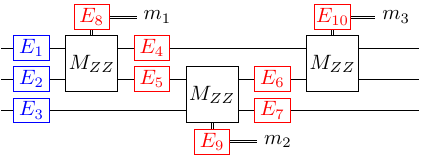}
\label{eq:condensed_rep_code_measurement}.
\end{equation}
\triangleqed
\end{example}

An error $\mathbf{e}$ in the phenomenological noise model consists of an input error $\mathbf{e}_\text{input}$ and an internal error $\mathbf{e}_\text{internal}$.
It leads to a propagated $\mathbf{e}_{\text{propagated}}$, a binary vector of length $n$.
An entry $j$ of $\mathbf{e}_{\text{propagated}}$ is $1$ if $\mathbf{e}_\text{input}$ and $\mathbf{e}_\text{internal}$ together contain an odd number of $X$ errors acting on qubit $j$.
We denote the weight of the propagated error as $\abs{\mathbf{e}_{\text{propagated}}}$.
For each measurement in a schedule, we introduce a detector whose value depends solely on its outcome.
As before, $\mathbf{s}$ denotes a syndrome. 
From the measurement outcomes obtained by performing a measurement schedule, the task of a decoder is to infer a correction.
From a syndrome $\mathbf{s}$, we denote the correction inferred by a decoder as $c(\mathbf{s})$.
An entry $j$ of $c(\mathbf{s})$ is $1$ if an $X$ operation is applied on qubit $j$ after $\mathbf{s}$ has been measured.
We refer to errors with the same syndrome as \textit{syndrome-consistent errors}.
We denote the set of all errors with syndrome $\mathbf{s}$ and weight at most $t$ as $\mathcal{E}(\mathbf{s}, t)$

Now we are ready to define fault-tolerant measurement schedules.

\begin{definition}[Fault-tolerant measurement schedule]
A measurement schedule is fault-tolerant to distance $d=2t+1$ if, $\forall \mathbf{s}:\exists c(\mathbf{s})$ such that $\forall \mathbf{e} \in \mathcal{E}(\mathbf{s},t)$
\begin{equation}
\abs{c(\mathbf{s}) \oplus \mathbf{e}_{\text{propagated}}} \leq \abs{\mathbf{e}_{\text{internal}}}.
\label{eq:FTEC}
\end{equation}

That is, for every syndrome $\mathbf{s}$, there exists a correction $c(\mathbf{s})$, such that the weight of the output error is at most the weight of the internal error for any syndrome consistent error with weight up to $t$. The output error is the result of applying the correction to the propagated error.
\end{definition}

\begin{example}[Showing a measurement schedule is fault-tolerant]
$\triangleright$
To show that a measurement schedule is fault-tolerant, we need to show that for every syndrome the fault-tolerance condition specified in Eq.~\eqref{eq:FTEC} is satisfied.
Using the detector error matrix is convenient, as for the case where $t=1$ its columns directly give all syndromes that need to be checked.
In this example, we first construct the detector error matrix, and then prove that the measurement schedule from the previous example is fault-tolerant.
If we choose the detectors
\begin{equation}
d_1 \colon m_1 = 0, \quad d_2 \colon m_2 = 0, \quad d_3 \colon m_1 \oplus m_3 = 0,
\end{equation}
we get the detector error matrix

\begin{equation}
\begin{bmatrix}
1 & 1 & 0 & 0 & 0 & 0 & 0 & 1 & 0 & 0\\
0 & 1 & 1 & 0 & 1 & 0 & 0 & 0 & 1 & 0\\
0 & 0 & 0 & 1 & 1 & 1 & 0 & 1 & 0 & 1\\
\end{bmatrix}.
\end{equation}

The second column of the matrix is the syndrome of the input error $E_2$.
Because the column is unique, $E_2$ has a unique syndrome, i.e.,  $\mathcal{E}([1,1,0], 1) = \{E_2\}$.
The propagated error corresponding to $E_2$ is $[0,1,0]$.
Therefore if this syndrome occurs the correction $c([1,1,0]) = [0,1,0]$ can be applied.
This correction can be implemented by applying an $X$ gate on the second qubit after the measurement schedule.
It satisfies Eq.~\eqref{eq:FTEC}, because $\forall \mathbf{e} \in \mathcal{E}([1,1,0], 1)$, \begin{equation}\abs{c(\mathbf{s}) \oplus \mathbf{e}_{\text{propagated}}} = \abs{[0,1,0] \oplus [0,1,0]} =  0.
\end{equation}
Similarly, the first column is unique, so if $E_1$ occurs the correction $c([1,0,0]) = [1,0,0]$ can be applied to satisfy Eq.~\eqref{eq:FTEC}. 
The input error on the 3rd qubit does not yield a unique syndrome, as $\mathcal{E}([0,1,0], 1) = \{E_3, E_9\}$.
For this syndrome, 
applying correction $X_3$ works: in case $E_9$ occurred, $\abs{c([0,1,0] \oplus \mathbf{e}_\text{propagated})} = 1 = \abs{\mathbf{e}_\text{internal}}$.
If $E_3$ occurred, $\abs{c([0,1,0] \oplus \mathbf{e}_\text{propagated})} = 0$.

For all other syndromes all syndrome consistent errors are internal errors, for which Eq.~\eqref{eq:FTEC} is satisfied if no correction is applied.
We have shown that for every possible syndrome the fault-tolerance condition is satisfied and, therefore, 
this measurement schedule is fault-tolerant with $t=1$. \triangleqed
\end{example}

In the previous example, we have seen that we can check if a measurement schedule is fault-tolerant by looking at the detector error matrix.
In the next example we will show how we can add measurements to a non-fault-tolerant measurement schedule in order to make it fault-tolerant.

\begin{example}[Designing a distance 3 measurement schedule] $\triangleright$
\label{example:distance_3_color_code_measurement_schedule}
\textit{Two-dimensional (2D) 
color codes} \cite{kesselring2022anyon,bombin2006topological} are topological QEC codes that can be defined on a lattice.
Vertices correspond to qubits and each plaquette corresponds to an $X$ and $Z$ stabilizer on the qubits that surround it.
In this example, we will design a measurement schedule for a color code with parameters $\left[12,2,3\right]$:

\begin{equation}
\begin{split}
\tikzfig{colour_code_circuits/12_2_3/q_labels}{0.8}   
\end{split}.
\end{equation}

Here, 
we have labeled the data qubits. 
We would like to find a measurement schedule using as few measurements as possible.
We are going to use the four weight 4 $Z$ \textit{plaquette stabilizers} (the green and blue plaquettes) and one weight 6 $Z$ plaquette stabilizer (the red plaquette).
We start with a schedule measuring the 5 plaquette stabilizers of the code in as few steps as possible,
\begin{equation}
\begin{split}
\tikzfig{colour\_code\_circuits/12\_2\_3/step\_1}{0.4}\;
\tikzfig{colour\_code\_circuits/12\_2\_3/step\_2}{0.4}\;
\tikzfig{colour\_code\_circuits/12\_2\_3/step\_3}{0.4}\\
\end{split}.
\end{equation}

The entire detector error matrix is too large to include here.
This schedule is not fault-tolerant. 
Consider for example an input error on qubit 10 
\begin{equation}
\begin{split}
\tikzfig{colour\_code\_circuits/12\_2\_3/step\_1\_input\_error}{0.4}\;
\tikzfig{colour\_code\_circuits/12\_2\_3/step\_2\_input\_error}{0.4}\;
\tikzfig{colour\_code\_circuits/12\_2\_3/step\_3\_input\_error}{0.4}
\end{split}
\end{equation}
and an internal error on qubit 6
\begin{equation}
\begin{split}
\tikzfig{colour\_code\_circuits/12\_2\_3/step\_1}{0.4}\;
\tikzfig{colour\_code\_circuits/12\_2\_3/step\_2\_internal\_error}{0.4}\;
\tikzfig{colour\_code\_circuits/12\_2\_3/step\_3\_input\_error}{0.4}.
\end{split}
\end{equation}

Because these two errors violate the same two detectors, but have different propagated errors, Eq.~\eqref{eq:FTEC} cannot be satisfied.
A computer program can be used to find all pairs of columns that are equal in the detector error matrix.
There are four pairs; the columns corresponding to input errors on qubits 9, 10, 11 and 12 are equal to the columns corresponding to internal errors before the second measurement round on qubits 5, 6, 7 and 8 respectively.
To make the schedule fault-tolerant we need to measure stabilizers that differentiate the errors in each of the four pairs.
This can be done by repeating the first measurement round consisting of measuring the two blue plaquettes at the end of the schedule.
Alternatively, a single measurement, namely the product of the two blue plaquettes, would also be sufficient.
\triangleqed
\end{example}

Above we have shown how to design a distance 3 measurement schedule.
Previous work shows that one does not have to use the detector error model framework to design measurement schedules \cite{delfosse2020short,delfosse2021beyond}. 
But compared to previous methods, the detector error model framework has two main advantages:
\begin{itemize}
    \item It treats measurement errors and data qubit errors on the same footing, which is useful for designing higher-distance measurement schedules. 
    \item Because equivalent errors are merged, the number of errors that need to be checked is minimized. This was, for example, done for the circuit in Eq.~\eqref{eq:rep_code_measurement_schedule_initial}, which resulted in the circuit in Eq.~\eqref{eq:condensed_rep_code_measurement}.
\end{itemize}
In Appendix \ref{sec:distance_5_measurement_schedules}, we discuss measurement schedules for the [17,1,5] and [19,1,5] color codes.

\section{Designing a fault-tolerant logical measurement procedure}
\label{sec:procedure}

In the previous section, we focused on measurement schedules used in fault-tolerant QEC procedures. In practice, to perform computations on logical qubits, additional procedures are necessary, including logical state preparation, logical gates, and logical measurements. In this section, we employ the framework of detector error models to design a fault-tolerant logical measurement procedure.
Our focus is on non-destructive logical measurement procedures that measure a single logical operator. 

\subsection{Logical measurement procedure}

To extract information from a noisy device, a procedure is needed that returns the correct measurement outcome as long as only a limited number of errors occur.

\begin{definition}{(Fault-tolerant logical measurement procedure \cite{gottesman2010introduction})}
A procedure for measuring a single logical operator is fault-tolerant to distance $d=2t+1$ if provided the number of input errors plus internal errors is at most $t$, the measurement outcome is correct and the output error weight is at most $t$.
\label{definition:fault_tolerant_logical_measurement_protocol}
\end{definition}

Naively, one might measure a logical operator via a fault-tolerant syndrome extraction circuit, as introduced in Definition \ref{def:FT_syndrome_extraction}.
We refer to such a one-off measurement of a logical operator as a \textit{logical measurement (LM) component}.
However, executing a single LM component does not constitute even a distance 3 fault-tolerant logical measurement protocol, as a single error can alter the measurement outcome.
Additionally, repeating the LM component multiple times does not increase the distance, as an input error could flip the measurement outcome of all LM components.

To satisfy Definition \ref{definition:fault_tolerant_logical_measurement_protocol}, error correction needs to be incorporated into a procedure. 
In Section 7.4 of Ref.~\cite{aliferis2005quantum} a distance 3 procedure incorporating error correction was introduced.
We refer to this procedure as the \emph{AGP procedure} -- after the initials of the authors of the original work -- and prove that it is fault-tolerant below.

\subsection{AGP procedure}
\label{sec:proof_correctness_known}

In addition to LM components, the AGP procedure incorporates \emph{quantum error correction} (QEC) components. 
A QEC component is a fault-tolerant QEC procedure, as in Definition \ref{definition:FTEC_protocol}.
The AGP procedure consists of consecutive application of five components: LM, QEC, LM, QEC, LM.  
We will prove the fault-tolerance of this procedure using detector error models.

The noise model that needs to be considered in this case is much simpler than the circuit-level and phenomenological noise models from previous sections.
We have a guarantee on the effect of any internal error and any input error on components, allowing us to treat all input and internal errors as equivalent.
Therefore, in the case of distance $3$ fault tolerance, we may assume that there is only a single possible input error and a single possible internal error that can occur at each component.

We show that the distance $3$ AGP procedure yields the correct measurement outcome as long as only a single error occurs. 
This is done by checking that there is no pair of errors that flips the observable without violating any detectors.
We represent the procedure as a circuit with separate components of the procedure drawn as gates. 
The detailed circuits implementing each component do not matter, as long as they satisfy the respective fault tolerance requirements. 

\begin{center}
\begin{equation}
\includegraphics[valign=c, width=0.85\columnwidth]{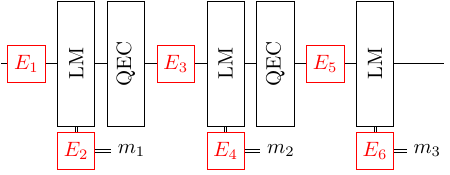}.
\end{equation}
\label{eq:d_3_AGP}
\end{center}
Here, we only depict errors that can affect measurements $m_1, m_2$, and $m_3$, as those are the only errors requiring consideration.
The red boxes labeled $E_1, E_3$ and $E_5$ represent potential locations for single-qubit Pauli errors on data qubits that anti-commute with the measured logical operator.
Similarly, the red boxes labeled $E_2, E_4$ and $E_6$ represent measurement errors that may occur. 

If no errors occur, the measurement results of all LM components should be identical. Let us, therefore, define the two detectors
\begin{equation}
d_1: m_1 \oplus m_2 = 0, \quad d_2: m_2 \oplus m_3 = 0,
\end{equation}
where $m_i$ denotes the measurement result of the $i$-th LM component. 
We choose our observable 
\begin{equation}
o_1: m_3
\end{equation}
to only depend on the final LM component.
The Tanner graph of the procedure can be depicted as
\begin{center}
\begin{equation}
\includegraphics[valign=c]{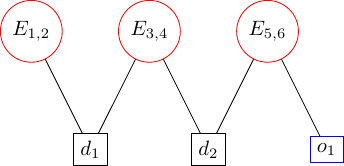}.
\end{equation}
\end{center}
From this Tanner graph, one can observe that there is no combination of two errors that flips $o_1$ but does not violate any detectors. 
Therefore, the correct value of $o_1$ can be obtained if at most one error occurs. 

A natural generalization of the distance 3 AGP procedure is fault-tolerant to higher distance: 
\begin{theorem}[Fault tolerance of the AGP procedure]
\label{theorem:AGP_procedure_is_fault-tolerent}
The distance $d$ generalization of the AGP procedure, LM followed by $d-1$ repetitions of the sequence QEC, LM, is a fault-tolerant logical measurement procedure.
\end{theorem}
We prove this in Appendix \ref{sec:proof_AGP_d}.

\subsection{Procedure using fault-tolerant error detection}

To generate a shorter fault-tolerant logical measurement procedure, we introduce a new component, the \emph{error detection} (QED) component.
We define a QED component to be a circuit that performs \emph{fault-tolerant error detection}.

\begin{definition}{(Fault-tolerant error detection procedure)} 
An error detection procedure is fault-tolerant to distance $d=2t+1$, if provided the number of input error is at most $t$ and there are no internal errors, it detects that an error has occurred.  
The output error weight is at most the number of internal plus input errors.
\end{definition}

Note that according to this definition, the detection of input errors is only required in the absence of internal errors.
Additionally, our definition requires only the detection of up to $t$ input errors.
This criterion is considerably less restrictive compared to that of a fault-tolerant error correction procedure, where $2t$ input errors need to be detected.

\begin{theorem}[Fault-tolerant measurement using QED]
\label{thm:fault_tolerance_distance_d_new_procedure}
Distance $d$ fault-tolerant measurement of a logical operator can be implemented using the following sequence; LM followed by $(d-1)/2$ iterations of the sequence: QED, LM, QEC, LM.
\end{theorem}

We again first prove the correctness of the distance 3 version of the new procedure. 
Compared to the distance 3 AGP procedure there is only one difference: the first QEC component is replaced with a QED component
\begin{equation}
\includegraphics[valign=c, width=0.85\columnwidth]{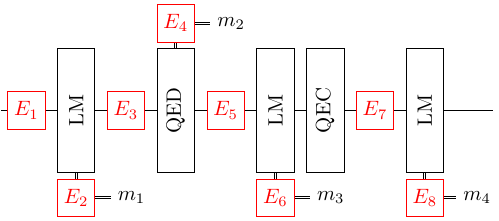}.
\end{equation}
We define three detectors and one observable as 
\begin{equation}
\begin{alignedat}{2}
d_1 &\colon m_1 \oplus m_3 = 0 ,\quad d_2 &&\colon m_2 = 0 ,\\
d_3 &\colon m_3 \oplus m_4 = 0 ,\quad o_1 &&\colon m_4 .
\end{alignedat}
\end{equation}
Here, $m_i \in \{1, 3, 4\}$ denotes the measurement result of the $i$-th LM component, and $m_2$ indicates whether the QED component has detected an error. 
The Tanner graph of this procedure shows that there is no combination of two errors the flips $o_1$ without violating any detector
\begin{center}
\begin{equation}
\includegraphics[valign=c]{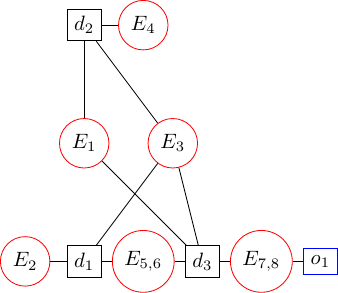}.
\end{equation}
\end{center}
In Appendix \ref{subsec:proof_new_d}, 
we give the full proof of Theorem \ref{thm:fault_tolerance_distance_d_new_procedure}.

\section{Conclusions and outlook}
\label{sec:conclusion}

In this work, we provide a pedagogical introduction to detector error models.
Currently, the QEC community lacks a standardized format for conveying information about a circuit to a decoder.
Because a detector error model contains all the information necessary to decode any Clifford circuit, the framework is ideally suited to fill this void.
The programming package \emph{Stim}, in which detector error models were originally proposed, can help facilitate this transition \cite{gidney2021stim}.
Our pedagogical introduction to detector error models aims to promote their widespread use.

In addition to the pedagogical introduction, we formalize how to design fault-tolerant circuits at various levels of abstraction using detector error models.
At the lowest level of abstraction in fault-tolerant circuit design, we constructed syndrome extraction circuits for the surface code that can correct twice as many measurement errors as conventional circuits.
We benchmarked our circuits by analyzing the logical error rates in memory and stability experiments under a measurement-biased noise model and a superconducting-inspired noise model. 
The proposed circuits demonstrate greater efficiency for circuit-level noise with higher measurement bias.
In the future, it would be interesting to apply our syndrome extraction circuits to other codes.

At a higher level of abstraction, we developed a systematic method to design measurement schedules.
We applied this method to design measurement schedules for distance 3 and distance 5 color codes.
For the triangular distance 5 color code on a hexagonal lattice, our measurement schedules are shorter than previously known schedules.
In the future, addressing the open problems involving measurement schedules given in Ref.\ \cite{delfosse2020short}, would be of interest. 
Using the detector error model framework may facilitate progress on these problems.

At the highest level of abstraction, we have designed a fault-tolerant logical measurement procedure for arbitrary codes.
The proposed procedure utilizes error detection to reduce time overhead.
Future investigations could focus on the logical error rate of this procedure. 
Additionally, using the detector error model framework may be beneficial in designing procedures for measuring multiple logical operators or operators acting on multiple logical qubits.

Throughout this work, our focus in designing fault-tolerant circuits has been on maximizing the number of errors the circuits are resilient to. 
While the number of errors a circuit can tolerate is a key measure of its error-correcting capabilities, it is not the only important metric. 
Different codes with the same distance can exhibit varying logical error rates, and similarly, different circuits implementing the same code can also have different logical error rates. 
An intriguing research avenue would be to explore how the properties of a detector error model influence the logical error rate of a circuit.

\section*{Acknowledgements}
We thank Craig Gidney for open-sourcing Stim and answering questions on the Quantum Computing StackExchange.
We thank Earl Campbell for his lectures on detector error models at the IBM Summer School.
We thank Lennart Bittel and Ben Criger for insightful discussions.
We thank Armanda Quintavalle, Tamas Noszko, and Julio Magdelena for giving helpful feedback on an earlier version of this work. This work has been supported by the BMBF (RealistiQ, QSolid, MuniQC-Atoms), the DFG (CRC 183), the Munich Quantum Valley (K-8), the Quantum Flagship (Millenion),
Berlin Quantum, and the Einstein Foundation (Einstein Research Unit on Quantum Devices).

\section*{Author contributions}
PJD, ATT, and AGB formalized the detector error model framework as presented here. 
PJD designed the majority of the circuits and performed numerics.
PJD wrote the manuscript with substantial input and feedback from all authors.

\bibliographystyle{quantum}
\bibliography{citations}

\appendix

\section{Circuit-level noise simulation details}
\label{sec:simulation_details}

\subsection{Simulation setup}

The most efficient known way to perform fault-tolerant computations using surface codes on a 2D grid of qubits with nearest-neighbor connectivity is lattice surgery \cite{horsman2012surface, litinski2019game}.
In computations using lattice surgery, there are two types of logical errors that can occur; space-like and time-like logical errors \cite{PRXQuantum.3.010331, higgott2023improved}. 
A logical $\ket{0}_{L}$ memory experiment can be used to calculate the error rate of a space-like logical $X$ error. 
A logical $Z$-stability experiment can be used to calculate the error rate of a time-like logical error flipping a computational basis measurement result \cite{gidney2022stability}.

\begin{table}
    \centering
    \begin{tabular}{|p{10mm}|p{33mm}|p{10mm}|p{10mm}|}
        \hline
        Gate & Noise applied after gate & SDMB model & SI model \\
        \hline \hline
        H & $X,Y,Z$ & $p$ & $p/10$ \\
        \hline
        CNOT & $\{I,X,Y,Z\}^{\times 2} \backslash \{II\}$ & $p$ & $p$  \\
        \hline
        Init. $\ket{0}$ & $X$ & $p$ & $2p$  \\
        \hline
        Meas. $Z$ & Bit-flip & $\eta \times p$ & $5p$  \\
        \hline
        Idle & $X,Y,Z$ & $p$ & $p/10$ \\
        \hline
        Res. idle* & $X,Y,Z$ & 0 & $2p$ \\
        \hline
    \end{tabular}
    \caption{Exact details of the noise models used for numerical simulations. The leftmost column contains all types of operations that are applied in the circuits. The second column lists which errors can occur after a gate. If multiple errors are given, one is chosen uniformly at random. In the third and fourth column, the probabilities are given with which an error occurs in the two noise models.*Resonator idle refers to a circuit location during which a qubit is not measured or reset in a time step during which other qubits are being measured or reset.}
    \label{tab:noise_table}
\end{table}

To benchmark the syndrome extraction circuits with two auxiliary qubits, we simulate memory and stability experiments subject to two different circuit-level noise models.  
We refer to the first model as \emph{standard depolarizing measurement-biased} (SDMB) noise. 
In SDMB, all gate errors happen with probability $p$ and measurement errors happen with probability $\eta \times p$, where $\eta$ quantifies the bias of the noise model towards measurement errors.
The second noise model is based on the \emph{superconducting inspired} (SI) noise model from Ref.~\cite{gidney2022benchmarking}.
Our version is not fully equivalent as our circuits use CNOTs instead of CZs.
For details on the noise models see Table \ref{tab:noise_table}. A more realistic comparison could take hardware specific aspects such as qubit connectivity into account, but this is beyond the scope of the present analysis.

The goal of designing QEC circuits is to reach a low logical error rate using as few qubits as possible.
To this end, one metric used is the
\textit{teraquop footprint} of a circuit, which represents the number of physical qubits required to reach a logical error rate of $10^{-12}$, the so called \textit{teraquop regime} \cite{gidney2021fault}.
To determine the logical error rate at a fixed distance and fixed physical error rate we perform Monte Carlo simulations. 
These involve sampling from the noise model and decoding the resulting errors using the open-source \emph{PyMatching} decoder, which implements a minimum-weight perfect matching algorithm \cite{higgott2022pymatching}. 

To then determine the teraquop footprint at a fixed physical error rate we perform these Monte Carlo simulations for various circuit distances $d$.
For each distance we generate $10^7$ samples or continue until $10^3$ logical errors are observed, whichever occurs first. 
From these simulations, we obtain a curve showing how the logical error rate scales with distance.
The distance required to reach the teraquop regime is identified by locating the point at which this curve intersects with the logical error rate of $10^{-12}$. Once the required distance is determined, we can calculate the corresponding teraquop footprint.
For single-measurement syndrome extraction, the number of qubits scales as $2d^2 - 1$, whereas for double-measurement syndrome extraction, it scales as $3d^2 - 2$.
We repeat this entire procedure for a range of different physical error rates.

\subsection{Results for measurement-biased noise}

\begin{figure}
\centering\includegraphics[width=\columnwidth]{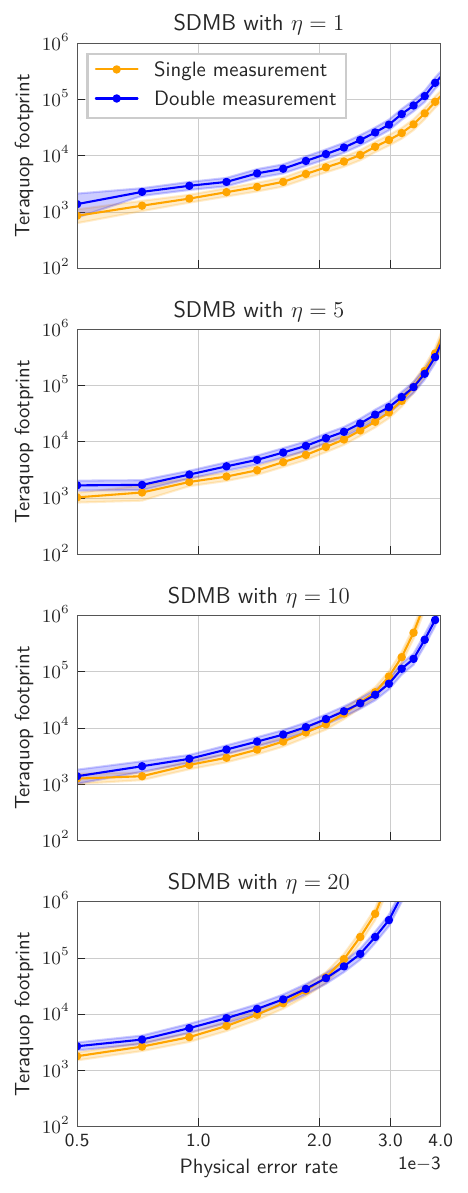}
    \caption{Teraquop plots obtained using memory experiments with the SDMB noise model for measurement biases $\eta = 1, 5, 10, 20$.}
\label{fig:memory_teraquop_plot}
\end{figure}

\begin{figure}
\centering\includegraphics[width=\columnwidth]{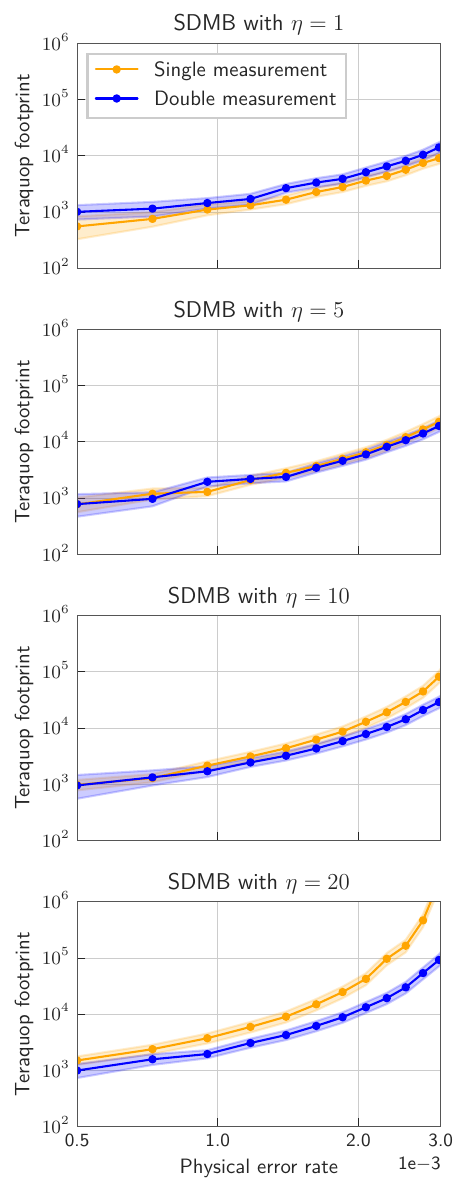}
    \caption{Teraquop plots obtained using stability experiments with the SDMB noise model for measurement biases $\eta = 1, 5, 10, 20$.}
    \label{fig:stability_teraquop_plot}
\end{figure}

The memory teraquop plots are shown in Fig.~\ref{fig:memory_teraquop_plot}.
To generate these plots we have simulated memory experiments with 4 different measurement biases $\eta = 1, 5, 10, 20$ using distances $d = 5, 7, 9, 11, 13$.
From the teraquop plots we draw two main conclusions.
First, the space-like logical error rate is relatively immune to increased measurement noise.
At a physical error rate of $10^{-3}$, the distance required to reach the teraquop regime only scales by a factor of approximately 2 when increasing the measurement bias $\eta$ from 1 to 20 for the single-measurement syndrome extraction circuits and a factor of approximately 1.5 for the double-measurement syndrome extraction circuits.
Second, there is a noise regime where double-measurement syndrome extraction is more efficient in terms of qubit count. 
Specifically, for $\eta=5$, $\eta=10$ and $\eta=20$ this regime is at physical error rates above $3.5 \times 10^{-3}, 2.5 \times 10^{-3}$ and $2 \times 10^{-3}$ respectively.

For stability experiments, we ran $d$ rounds of syndrome extraction with surface code patch diameter $d$, for $d \in \{4,8,10,12,14\}$.
The teraquop footprint plots are shown in Fig.~\ref{fig:stability_teraquop_plot}.
These plots show that for $\eta=10$ and $\eta=20$ the teraquop footprint of the double syndrome extraction circuits is lower.

\subsection{Results for superconducting inspired noise}

The teraquop plots for the SI1000 noise model are shown in Fig.~\ref{fig:SI1000_teraquop}.
For the SI1000 noise model, the results indicate that the measurement noise bias is not strong enough for the proposed double-measurement syndrome extraction to be more efficient. 
In the memory experiment plot it can be seen that for $p > 2.5 \times 10^{-3}$ the double-measurement syndrome extraction is more efficient. 
For lower physical error rates the double measurement syndrome extraction circuit uses approximately twice as many qubits. 
The stability plot shows that the footprint of the double-measurement syndrome extraction circuit is smaller for $ p > 2 \times 10^{-3}$ and approximately equal for lower error rates.

\newpage 

\begin{figure}
    \centering
    \includegraphics[width=\columnwidth]{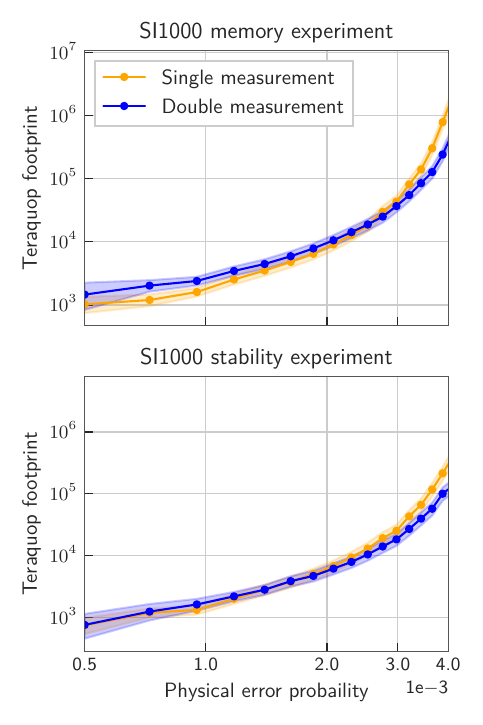}
    \caption{Teraquop plots obtained using memory and stability experiments with the SI1000 noise model.}
    \label{fig:SI1000_teraquop}
\end{figure}

\section{Distance 5 measurement schedules}
\label{sec:distance_5_measurement_schedules}
In the examples in Section \ref{sec:measurement_schedule}, we showed 
that it is relatively simple to design distance 3 measurement schedules.
Unfortunately, for higher distances it gets more complicated to design measurement schedules.
Only a handful of measurement schedules are known for distance 4 and 5 codes \cite{prabhu2021distance, delfosse2020short}.
In this section, we make a contribution to the collection of known measurement schedules by designing distance 5 schedules for color codes.
It has been shown that, for color codes of distance $d$, schedules using $O(d \log d)$ measurements exist \cite{delfosse2021beyond}, but constructing explicit schedules remains an open problem \cite{delfosse2020short}.

One type of systematic approach to designing measurement schedules is to have an initial non-fault-tolerant schedule and then add measurements to make it fault-tolerant.
Indeed this is what we did for the $[12,2,3]$ color code in Example \ref{example:distance_3_color_code_measurement_schedule}.
We started with a non-fault-tolerant schedule that measured a generating set of stabilizers, and then checked which stabilizers to add.
To apply this approach to distance 5 codes, a systematic way to generate an initial measurement schedule is needed.
A schedule that consists of measuring a generating set of stabilizers of a distance 5 code seems like a good starting point.
But, it turns out not to be.
The reason is that most pairs of errors cannot be corrected by a measurement schedule that measures a generating set, so figuring out the minimal set of stabilizers to add is difficult.
What about measuring a generating set of stabilizers twice?
This does enable finding a measurement schedule.
We showcase this approach in the example below.

\begin{example}[Distance 5 repetition code] $\triangleright$
A generating set of stabilizers for the distance 5 repetition code is $\{Z_1 Z_2, Z_2 Z_3, Z_3 Z_4, Z_4 Z_5\}$.
These stabilizers can each be measured twice in four rounds. 
This is not fault-tolerant, due to the existence of syndromes $\mathbf{s}$ for which $\forall \mathbf{e} \in \mathcal{E}(\mathbf{s},2)$
\begin{equation}
\abs{c(\mathbf{s}) \oplus \mathbf{e}_{\text{propagated}}} \leq \abs{\mathbf{e}_{\text{internal}}} \quad 
\label{eq:FTEC_again}
\end{equation}    
can not be satisfied.
Consider the pair of input errors $E_1$ and $E_2$ and the pair of internal errors $E_3$ and $E_4$ drawn below,
\begin{equation}
\includegraphics[valign=c, width=0.85\columnwidth]{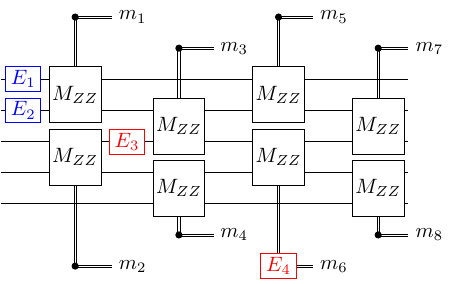}.
\end{equation}

These pairs of errors have the same syndrome, and a correction $c(\mathbf{s})$ satisfying Eq.~\eqref{eq:FTEC_again} does not exist. 
The measurement schedule can be made fault-tolerant by adding a fifth round.
One option for the fifth round is to measure $Z_1 Z_2$ and $Z_3 Z_4$.
\label{ex:5_rounds}
\triangleqed
\end{example}

A natural question to ask is whether the approach used in the previous example leads to the shortest possible measurement schedule for a given code.
The next example shows that the answer is no, as we find a measurement schedule for the distance 5 repetition code that uses only 4 rounds. 

\begin{example}[4 round schedule for the distance 5 repetition code] $\triangleright$
The following measurement schedule, which measures $Z_1 Z_5$ and $Z_3 Z_4$ in the third round, is fault-tolerant.
We numerically confirmed its fault tolerance by iterating over all error combinations of weight not exceeding 2.
\begin{equation}
\includegraphics[valign=c, width=0.85\columnwidth]{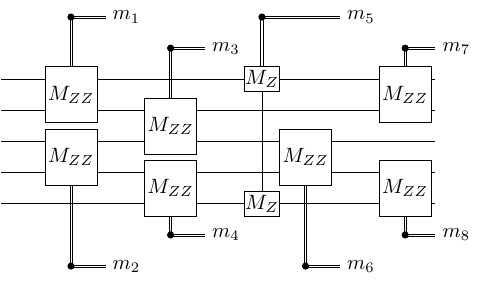}.
\end{equation}
\label{ex:4_rounds}
\triangleqed
\end{example}

Now that we have designed some measurement schedules for the repetition code, we shift focus to codes that can correct $X$ and $Z$ errors.
We focus on the distance 5 color code, whose small size makes its interesting for near-term experiments.
In particular, the [17,1,5] color code is the smallest distance 5 CSS code we are aware of.
In previous work, a measurement schedule for the [17,1,5] color code that uses 20 measurements in 9 time steps has been given in Ref.\ \cite{delfosse2020short}.
We show four different practical measurement schedules in the example below, including a schedule that employs only 17 measurements and 7 time steps.

\begin{example}[Schedule for the 17 qubit code] $\triangleright$
\label{example:17_1_5_short}
The measurement schedule below is the shortest one we have found for the [17,1,5] code. 
It uses 17 measurements in 7 steps. 
At timesteps 4, 5 and 6, products of plaquette stabilizers are measured.
The qubits involved in the measured products are represented by the grey vertices.

\begin{equation}
\begin{split}
\tikzfig{colour_code_circuits/17_1_5_non_local/step_1}{0.2}\;
\tikzfig{colour_code_circuits/17_1_5_non_local/step_2}{0.2}\;
\tikzfig{colour_code_circuits/17_1_5_non_local/step_3}{0.2}\;
\tikzfig{colour_code_circuits/17_1_5_non_local/step_4}{0.2}\\
\tikzfig{colour_code_circuits/17_1_5_non_local/step_5}{0.2}\;
\tikzfig{colour_code_circuits/17_1_5_non_local/step_6}{0.2}\;
\tikzfig{colour_code_circuits/17_1_5_non_local/step_7}{0.2}\hspace{1cm}
\end{split}
\end{equation}

Instead of a systematic approach, we took a more heuristic approach to find this schedule.
Specifically, this measurement schedule is inspired by the schedule in Example \ref{ex:4_rounds}.
If one looks at the 5 qubits on any of the three boundaries, the measurement schedule on just those data qubits looks similar to the four round measurement schedule of the repetition code. 
For example, consider the 5 qubits on the bottom boundary of the code. 
On these qubits, the measurements at steps 1, 3, 6, and 7 shown below are nearly the same as the measurements in the four round schedule of the repetition code.

\begin{equation}
\tikzfig{colour_code_circuits/17_1_5_non_local/sub_schedule_rep_code}{0.18}
\end{equation}
\triangleqed
\end{example}

\begin{example}[Local schedule for the 17 qubit code] $\triangleright$
For quantum computers with local connectivity, only local stabilizers can be measured. 
Below is a measurement schedule for the [17,1,5] code that only uses local stabilizers and uses 20 measurements in 8 steps.
This measurement schedule was inspired by the local measurement schedule for the repetition code in Example \ref{ex:5_rounds}.
\begin{equation}
\begin{split}
\tikzfig{colour_code_circuits/17_1_5_local_3_octagon/step_1}{0.2}\;
\tikzfig{colour_code_circuits/17_1_5_local_3_octagon/step_2}{0.2}\;
\tikzfig{colour_code_circuits/17_1_5_local_3_octagon/step_3}{0.2}\;
\tikzfig{colour_code_circuits/17_1_5_local_3_octagon/step_4}{0.2}\\
\tikzfig{colour_code_circuits/17_1_5_local_3_octagon/step_5}{0.2}\;
\tikzfig{colour_code_circuits/17_1_5_local_3_octagon/step_6}{0.2}\;
\tikzfig{colour_code_circuits/17_1_5_local_3_octagon/step_7}{0.2}\;
\tikzfig{colour_code_circuits/17_1_5_local_3_octagon/step_8}{0.2}
\end{split}
\end{equation}
\triangleqed
\end{example}

\begin{example}[Local schedule for the 17 qubit code with few high-weight measurements] $\triangleright$
In practice, it may be desirable to measure the weight 8 plaquette stabilizer of the [17,1,5] code as few times as possible.
Below we show a measurement schedule for the [17,1,5] code that uses 21 measurements in 8 steps, and measures the weight 8 stabilizer twice.
\begin{equation}
\begin{split}
\tikzfig{colour_code_circuits/17_1_5_local_2_octagon/step_1}{0.2}\;
\tikzfig{colour_code_circuits/17_1_5_local_2_octagon/step_2}{0.2}\;
\tikzfig{colour_code_circuits/17_1_5_local_2_octagon/step_3}{0.2}\;
\tikzfig{colour_code_circuits/17_1_5_local_2_octagon/step_4}{0.2}\\
\tikzfig{colour_code_circuits/17_1_5_local_2_octagon/step_5}{0.2}\;
\tikzfig{colour_code_circuits/17_1_5_local_2_octagon/step_6}{0.2}\;
\tikzfig{colour_code_circuits/17_1_5_local_2_octagon/step_7}{0.2}\;
\tikzfig{colour_code_circuits/17_1_5_local_2_octagon/step_8}{0.2}
\end{split}
\end{equation}
\triangleqed
\end{example}

\begin{example}[Schedule for the 19 qubit code]\label{example:19_1_5} $\triangleright$
Although the [19,1,5] uses two more data qubits, it may lead to lower overhead because it has weight 6 instead of weight 8 stabilizers.
Below, we present a measurement schedule comprising 19 measurements executed in 8 steps, with the maximum measurement weight being 6.
\begin{equation}
\begin{split}
\tikzfig{colour_code_circuits/19_1_5_non_local/step_1}{0.2}\;\;
\tikzfig{colour_code_circuits/19_1_5_non_local/step_2}{0.2}\;\;
\tikzfig{colour_code_circuits/19_1_5_non_local/step_3}{0.2}\;\;
\tikzfig{colour_code_circuits/19_1_5_non_local/step_4}{0.2}\\
\tikzfig{colour_code_circuits/19_1_5_non_local/step_5}{0.2}\;\;
\tikzfig{colour_code_circuits/19_1_5_non_local/step_6}{0.25}\;\;
\tikzfig{colour_code_circuits/19_1_5_non_local/step_7}{0.2}\;\;
\tikzfig{colour_code_circuits/19_1_5_non_local/step_8}{0.2}
\end{split}
\end{equation}
\triangleqed
\end{example}

\section{Justification of the phenomenological noise model}
\label{sec:justification_noise_model}

In Section \ref{sec:measurement_schedule} and Appendix \ref{sec:distance_5_measurement_schedules} we designed measurement schedules for correcting $X$ errors and measurement errors.
In this appendix we show how these schedules can be employed to create circuits resilient to circuit-level noise.
The resulting circuits are shown in Figure \ref{fig:qec_procedure_components}.

We begin by showing how to correct $XZ$ phenomenological noise, which includes single-qubit $X$ and $Z$ Pauli errors, as well as measurement errors. Following this, we show how to correct circuit-level noise.

\subsection{Correcting \texorpdfstring{$XZ$}{XZ} phenomenological noise}
The $Z$ measurement schedules in Section \ref{sec:measurement_schedule} were designed such that in combination with a decoder they satisfy;
if the sum of the number of $X_\text{input}$ errors and $X_\text{internal}$ phenomenological errors is at most $t$, then the output error's $X$ weight is at most the number of $X_\text{internal}$ phenomenological errors.
Here $X_\text{internal}$ errors can be measurement errors or $X$ data qubit errors.
Conveniently, the codes in Example \ref{example:distance_3_color_code_measurement_schedule} and Examples \ref{example:17_1_5_short}-\ref{example:19_1_5} are self-dual.
A consequence of this is that by copying the $Z$ measurement schedule and replacing the measurements of $Z$ operators by measurements of $X$ operators, we obtain an $X$ measurement schedule.
Naturally the fault-tolerant properties of the $Z$ measurement schedule carry over. 
We refer to a procedure which consist of an $X$ measurement schedule with a decoder followed by a $Z$ measurement schedule with a decoder as a CSS QEC procedure.
If a CSS QEC procedure is performed and the sum of $Z$ input errors and $Z$ internal errors ($Z$ data qubit errors and measurements errors on measurements of $X$ stabilizers) is at most $t$, then the output errors $Z$ weight is at most the number of $Z$ internal errors. 
The same holds for the output errors $X$ weight.

\subsection{Correcting circuit-level noise}

To make a fault-tolerant CSS QEC procedure resilient to circuit-level noise we replace the multi-qubit measurement gates with \textit{fault-tolerant syndrome extraction circuits}.

\begin{definition}[Fault-tolerant syndrome extraction circuit]
\label{def:FT_syndrome_extraction}
A syndrome extraction circuit $\mathfrak{C}$ is fault-tolerant if any $w$ circuit-level errors propagated to the end of $\mathfrak{C}$ results in an error with $X$ and $Z$ weight at most $w$.
\end{definition}

\begin{example}[Fault-tolerant syndrome extraction] $\triangleright$
An example of a fault-tolerant syndrome extraction circuit for measuring a weight 4 $Z$ operator is 
given by \footnote{using this circuit to perform syndrome extraction may not be practical for some architectures.}
\cite{shor1996fault}
\begin{equation}
\includegraphics[valign=c]{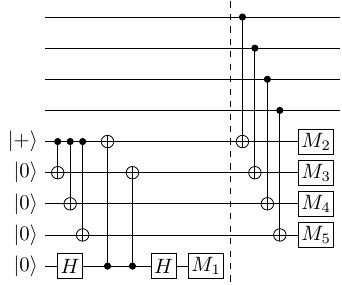}.
\label{eq:shor_syndrome_extraction_circ}
\end{equation}

Here, the gates after the dashed line are only performed if the first measurement has result $0$. 
If the measurement result is $1$, the circuit is started again. 
If we assume the top four qubits to be in the +1 eigenstate of the weight 4 $Z$ operator, we can define a detector $d_1 : m_2 \oplus m_3 \oplus m_4 \oplus m_5 = 0$.
\triangleqed
\label{example:ft_syndrome_extraction}
\end{example}

\begin{theorem}[Circuit-level noise resilience of CSS QEC procedure.]
If the measurements of a CSS QEC procedure are implemented using fault-tolerant syndrome extraction circuits, the resulting procedure satisfies Definition \ref{definition:FTEC_protocol} for circuit-level noise.
\label{theorem:circuit_level_noise_resilince}
\end{theorem}

Showing that the resulting circuit is resilient to all possible circuit-level errors sounds like a daunting task.
In this case, due to the notion of \textit{equivalent errors},  it is not.
We delay the proof of Theorem \ref{theorem:circuit_level_noise_resilince} until after we have defined equivalent errors and proven a useful lemma.

\begin{definition}[Equivalent errors]
Two sets of errors $\mathcal{E}_1$ and $\mathcal{E}_2$ that occur in Clifford circuits $\mathfrak{C}_1$ and $\mathfrak{C}_2$ with corresponding detector error matrices $H_1$ and $H_2$ are equivalent if:
\begin{itemize}
    \item $H_1 \mathbf{e}_1 = H_2 \mathbf{e}_2$, where $\mathbf{e}_1$ and $\mathbf{e}_2$ are the circuit vectors corresponding to $\mathcal{E}_1$ and $\mathcal{E}_2$.
    \item $\mathfrak{C}_1(\mathcal{E}_1) =  \mathfrak{C}_2(\mathcal{E}_2)$. $\mathfrak{C_i}(\mathcal{E}_j)$ denotes the operator found by propagating $\mathcal{E}_j$ to the end of $\mathfrak{C}_i$.
\end{itemize}
\label{def:equivalent_errors}
\end{definition}

That is to say, two sets of errors are equivalent if they have the same syndrome and result in the same propagated error.

\begin{example}[Equivalent errors]$\triangleright$
Consider $\mathcal{E}_1$ consisting of 2 circuit-level errors, a single-qubit $Z$ error and a two-qubit $Y \otimes X$ error occurring during the circuit $\mathfrak{C}_1$ from Example~\ref{example:ft_syndrome_extraction}

\begin{equation}
\includegraphics[valign=c, width=0.85\columnwidth]{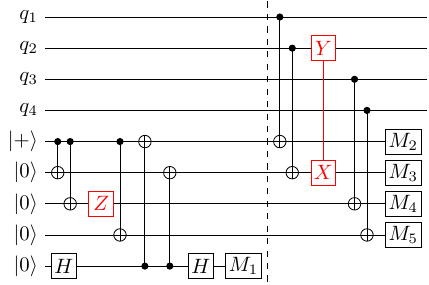}
\label{eq:example_error_shor}
\end{equation}

The $Y$ and $X$ errors are connected by a line, because $Y \otimes X$ after a CNOT is considered a single error in a circuit-level noise model.
The $Y \otimes X$ error flips the measurement $M_3$ and therefore violates the only detector $d_1 \colon m_2 \oplus m_3 \oplus m_4 \oplus m_5 = 0$. 
Thus $H_1 \mathbf{e}_1 = [1]$.
The error results in the propagated error $\mathfrak{C}(\mathcal{E}_1) = Y_{q_2} \otimes Z_{q_3}$ on the data qubits.

Now consider $\mathcal{E}_2$ consisting of  two $Z$ errors and one $X$ error occurring during the following circuit $\mathfrak{C}_2$
\begin{equation}
\includegraphics[valign=c]{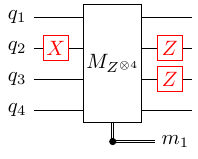}.
\end{equation}
If we assume the four qubits are in a $+1$ of $Z^{\otimes 4}$, we can define one detector: $d_1 \colon m_1 = 0$.
The errors flip the measurement and therefore violate the detector, i.e.,  $H_2 \mathbf{e}_2 = [1]$.
The errors propagate to $\mathfrak{C}_2(\mathcal{E}_2) = Y_{q_2} \otimes Z_{q_3}$.
Because $H_1 \mathbf{e}_1 = H_2 \mathbf{e}_2$ and $\mathfrak{C}_1 (\mathcal{E}_1) = \mathfrak{C}_2 (\mathcal{E}_2)$, $\mathcal{E}_1$ is equivalent to $\mathcal{E}_2$.
\triangleqed
\end{example}

We refer to a fault-tolerant syndrome extraction circuit as $\mathfrak{C}_\text{ft}$.
We refer to a syndrome extraction circuit that consists of a single multi-qubit measurement gate as $\mathfrak{C}_\text{pheno}$.
We refer to errors that can occur in the $XZ$ phenomenological noise model as \textit{phenomenological errors}.
It turns out that all circuit-level errors occurring during $\mathfrak{C}_\text{ft}$ are equivalent to all previously considered phenomenological errors occurring during $\mathfrak{C}_\text{pheno}$.

\begin{lemma}[All circuit-level errors]
Any set of $w$ circuit-level errors that can occur in $\mathfrak{C}_\text{ft}$ is equivalent to a set containing at most $w$ $X_\text{internal}$ phenomenological errors and at most $w$ $Z_\text{internal}$ phenomenological errors that can occur in $\mathfrak{C}_\text{pheno}$.
\label{lemma:equivalence}
\end{lemma}

\begin{proof}
To prove Lemma \ref{lemma:equivalence}, we need only demonstrate, due to linearity, that all possible single circuit-level errors are equivalent to at most a single $X_\text{internal}$ phenomenological error and a single $Z_\text{internal}$ phenomenological error. 
We categorize circuit-level errors occurring during $\mathfrak{C}_\text{ft}$ in two groups depending on the value of $H_{ft} \mathbf{e}$.
Let's start by considering errors for which $H_{ft} \mathbf{e} = [0]$.
Due to the property of the fault-tolerant syndrome extraction circuit, for any $\mathcal{E}$ the $X$ weight of $\mathfrak{C}_\text{ft}(\mathcal{E})$ is at most 1 and the $Z$ weight of $\mathfrak{C}_\text{ft}(\mathcal{E})$ is at most 1.
Therefore the errors in this group are equivalent to at most one $X$ and one $Z$ error happening on data qubits at the end of $\mathfrak{C}_\text{pheno}$. 

Below, the fault-tolerant syndrome extraction circuit on the left side\footnote{For brevity, instead of drawing the gates to the left of the dotted line in Eq.~(\ref{eq:shor_syndrome_extraction_circ}), we have written $\ket{GHZ}$, which denotes the state those gates generate.} contains an example of a circuit-level error for which $H_{ft} \mathbf{e} = [0]$. 
The circuit on the right side shows the phenomenological error which the circuit-level error is equivalent to
\begin{equation}
\includegraphics[valign=c, width=0.85\columnwidth]{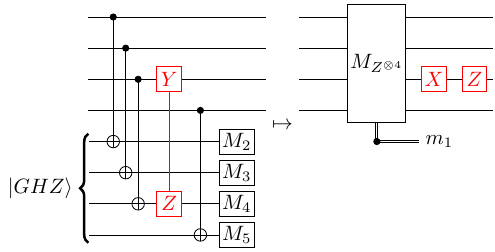}.
\end{equation}
The second group, comprising errors for which $H_{ft} \mathbf{e} = [1]$, is more complicated.
We split these errors into three subgroups:
\begin{enumerate}
\item Circuit-level errors for which $\mathfrak{C}_\text{ft}(\mathcal{E}) = I$.
These errors are equivalent to a measurement error in $\mathfrak{C}_\text{pheno}$.
An example of this case is shown below, where $E$ denotes a bit flip error on the measurement outcome
\begin{equation}
\includegraphics[valign=c, width=0.85\linewidth]{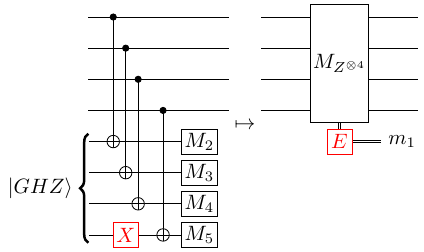}.
\end{equation}
\item Circuit-level errors for which $\mathfrak{C}_\text{ft}(\mathcal{E}) \neq I$ and $\mathfrak{C}_\text{ft}(\mathcal{E})$ anti-commutes with the measured operator.
If $\mathfrak{C}_\text{ft}(\mathcal{E}) = X_q$ or $\mathfrak{C}_\text{ft}(\mathcal{E}) = Z_q$, the error is equivalent to $X_q$ or $Z_q$, which are phenomenological errors, happening at the start of $\mathfrak{C}_\text{phen}$.
If $\mathfrak{C}_\text{ft}(\mathcal{E}) = Y_q$, 
the error is equivalent to the two errors $X_q$ and $Z_q$ happening at the start of $\mathfrak{C}_\text{phen}$. 
An example of this case is

\begin{equation}
\includegraphics[valign=c, width=0.85\linewidth]{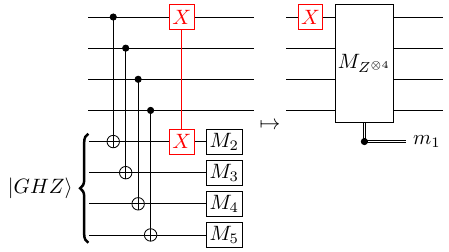}
\end{equation}

\item Circuit-level errors for which $\mathfrak{C}_\text{ft}(\mathcal{E}) \neq I$ and $\mathfrak{C}_\text{ft}(\mathcal{E})$ commutes with the measured operator.
If $\mathfrak{C}_\text{ft}(\mathcal{E}) = X_q$, $\mathcal{E}$ is equivalent to the following two phenomenological errors: $X_q$  at the end of $\mathfrak{C}_\text{pheno}$, which is an $X_\text{internal}$ error, and a measurement error on the measurement of an $X$ stabilizer, which is a $Z_\text{internal}$ error. Vice-versa if $\mathfrak{C}_\text{ft}(\mathcal{E}) = Z_q$. We do not consider $\mathfrak{C}_\text{ft}(\mathcal{E}) = Y_q$, as we are only considering syndrome extraction circuits measuring all $X$ or all $Z$ stabilizers here. Therefore if $\mathfrak{C}_\text{ft} = Y$ it anticommutes with the measured operator.
An example of this case is
\begin{equation}
\includegraphics[valign=c, width=0.85\linewidth]{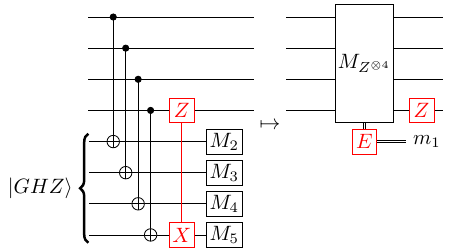}.
\end{equation}
\end{enumerate}
This concludes our proof as we have exhaustively shown that all possible single circuit-level errors are equivalent to at most a single $X$ internal phenomenological error and a single $Z$ internal phenomenological error. 
Now we are ready to bring everything together and prove Theorem \ref{theorem:circuit_level_noise_resilince}.
\end{proof}

\begin{proof}[Proof of Theorem \ref{theorem:circuit_level_noise_resilince}]
Lemma \ref{lemma:equivalence} states that all circuit-level errors are equivalent to a bounded number of errors in the XZ phenomenological noise model.
The bounded number of errors is exactly the amount for which it was shown in the previous section that the CSS QEC procedure is fault-tolerant.
Thus, the CSS QEC procedure constructed using measurement schedules and fault-tolerant syndrome extraction circuits is fault-tolerant for circuit-level noise.
\end{proof}

\section{Proofs of Theorems \ref{theorem:AGP_procedure_is_fault-tolerent} and \ref{thm:fault_tolerance_distance_d_new_procedure}}
\label{sec:comparison_logical_measurement}

\subsection{Proof of Theorem \ref{theorem:AGP_procedure_is_fault-tolerent}}
\label{sec:proof_AGP_d}

Before we prove Theorem \ref{theorem:AGP_procedure_is_fault-tolerent}, we prove the following intermediate Lemma.

\begin{lemma}[Lower bound to error number]
In the AGP procedure, it takes at least $2t+1$ errors to flip the observable without violating any detectors.
\label{lemma:agp_procedure_2t+1}
\end{lemma}
\begin{proof}
To prove this we split the circuit into 3 parts.
The first part consists of the first LM component.
The middle part consists of $d-2$ repetitions of a QEC component followed by a LM component.
The final part consists of a QEC component followed by a LM 
component.
The circuit representation  of the procedure is 
\begin{equation}
\includegraphics[valign=c, width=0.85\columnwidth]{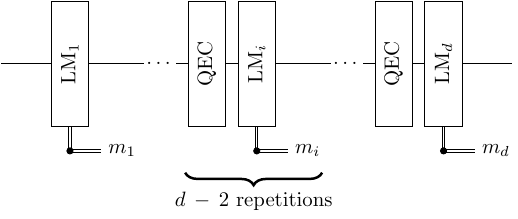}.
\end{equation}
There are $d-1$ detectors in the circuit, 
\begin{equation}
d_j: m_{j} \oplus m_{j+1} = 0, \quad j \in {1, \dots , d-1}
\end{equation}
and one observable, $o_1 : \text{m}_d$.

Now we will find the minimum number of errors that flips the observable without violating any detectors.
In the final part of the component, there are two choices for flipping the observable using as few errors as possible, either an input error before the component LM$_d$ or a measurement error during component LM$_d$
\begin{equation}
\includegraphics[valign=c]{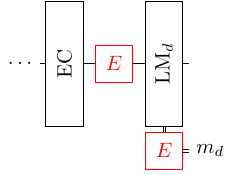}
\label{eq:errors_flipping_observable}.
\end{equation}
Both of these errors also violate the last detector $d_{d-1}$. 
To not violate this detector, an error needs to be used in the last repetition of the middle part.
But this error will violate the detector $d_{d-2}$. 
Therefore, an error will be needed in each of the $d-2$ rounds in the middle part of the procedure.
Finally, to not violate the first detector, an error will be needed in the first part of the procedure.
This gives a total of $2t+1$ errors needed to flip the observable without violating any detectors.
A Tanner graph of the procedure and the logical error we just described is 
\begin{equation}
\includegraphics[valign=c]{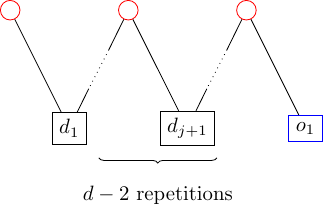}.
\end{equation}    
\end{proof}

Now we can prove Theorem \ref{theorem:AGP_procedure_is_fault-tolerent}.
\begin{proof}
According to Definition \ref{definition:fault_tolerant_logical_measurement_protocol}, for distance $d=2t+1$, two conditions must be satisfied; the outcome of the logical measurement should be correct if at most $t$ errors occur and the output error's weight should be at most $t$.
The latter condition is automatically satisfied due to the properties of the QEC and LM components.
The former condition is satisfied as we have proven Lemma \ref{lemma:agp_procedure_2t+1}, that states that it takes at least $2t+1$ errors to flip the observable without violating any detectors.
Because this takes $2t+1$ errors, $2t$ are detectable, and therefore $t$ errors are correctable.
Thus both conditions are satisfied, thereby proving Theorem \ref{theorem:AGP_procedure_is_fault-tolerent}.
\end{proof}

\subsection{Proof of Theorem \ref{thm:fault_tolerance_distance_d_new_procedure}}
\label{subsec:proof_new_d}

We will follow the structure of the proof in the previous section and again start by proving an auxiliary lemma.

\begin{lemma}[Second lower bound to error number]
In the proposed procedure using QED components, it takes at least $2t+1$ errors to flip the observable without violating any detectors.
\label{lemma:new_procedure_2t+1}
\end{lemma}

\begin{proof}
The circuit representation of the proposed procedure for arbitrary $d$ is
\begin{equation}
\includegraphics[valign=c, width=0.85\columnwidth]{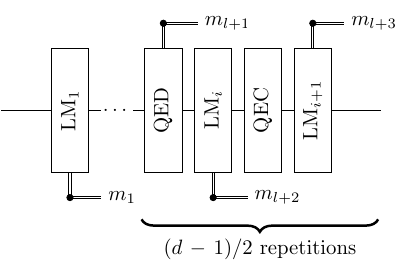}.
\label{eq:circuit_proposed_procedure}
\end{equation}
There are $(3d-3)/2$ detectors in the circuit,
\begin{align}
d_{l} \colon m_{l} \oplus m_{l+2} = 0 ,\nonumber\\
d_{l+1} \colon m_{l+1} = 0 \nonumber ,\\
d_{l+2} \colon m_{l+2} \oplus m_{l+3} = 0,\\
l \in 1,4, 7, \dots , (3d-7)/2, \nonumber
\end{align}
and one observable $o_1: m_{(3d-1)/2}$.

To create an error that flips the observable, again either a measurement error during the last LM component or an input error before the last LM component can be used.
Either of these errors violates the last detector. To not violate this detector, at least two errors are needed in each of the $(d-1)/2$ repetitions shown in Eq.~\eqref{eq:circuit_proposed_procedure}.
This can be seen by looking at a Tanner graph of the detector error matrix of the circuit. 
Below the circuit containing the individual errors that violate detectors and the corresponding Tanner graph is shown.

\begin{widetext}
\begin{equation}
\includegraphics[valign=c, width=0.9\linewidth]{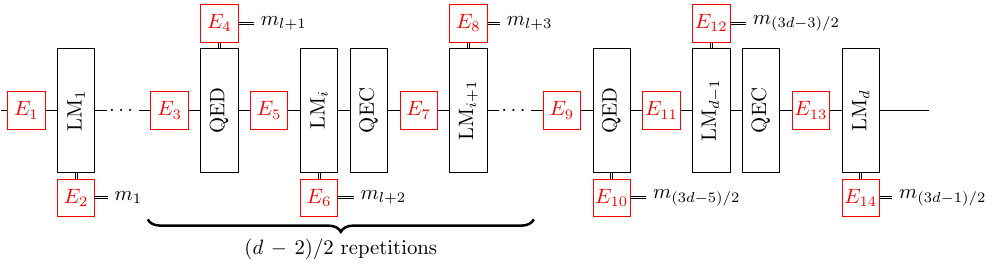}
\end{equation}
\begin{equation}
\includegraphics[valign=c, width=0.9\columnwidth]{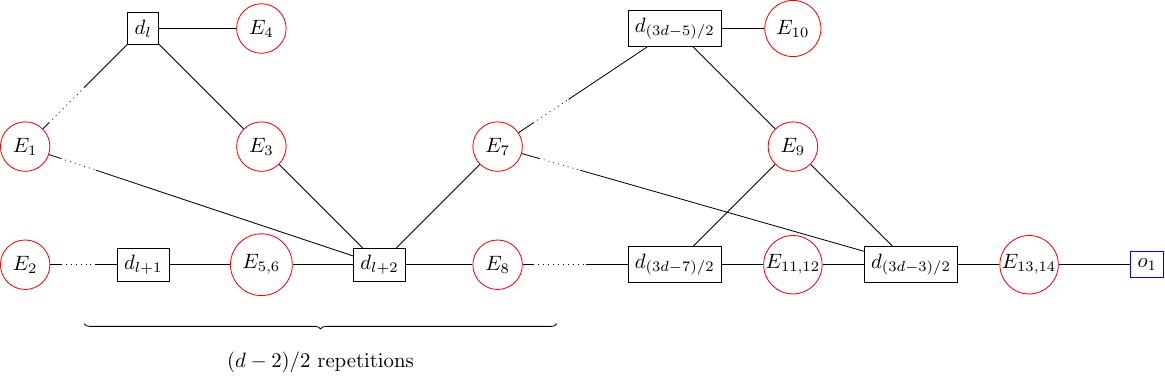}.
\end{equation}
\end{widetext}
We can conclude that $d$ errors are needed to flip the observable without violating any detectors. 
\end{proof}

Now our proof of Theorem \ref{thm:fault_tolerance_distance_d_new_procedure} is nearly identical to our proof of Theorem \ref{theorem:AGP_procedure_is_fault-tolerent}.

\begin{proof}
Again, two conditions must be satisfied; the outcome of the logical measurement should be correct if at most $t$ errors occur and the output error's weight should be at most $t$.
The latter condition is automatically satisfied due to the properties of the QEC, QED, and LM components.
It follows from Lemma \ref{lemma:new_procedure_2t+1} that the former condition is satisfied.
\end{proof}

\end{document}

%% file: colour_code_circuits/color_code.tikzstyles
\tikzstyle{new style 0}=[fill=black, draw=black, shape=circle, scale=0.3]
\tikzstyle{new style 1}=[fill={rgb,255: red,128; green,128; blue,128}, draw={rgb,255: red,128; green,128; blue,128}, shape=rectangle, scale=0.3]
\tikzstyle{new style 2}=[fill={rgb,255: red,0; green,128; blue,128}, draw={rgb,255: red,0; green,128; blue,128}, shape=circle]
\tikzstyle{new style 3}=[fill={rgb,255: red,191; green,0; blue,64}, draw={rgb,255: red,191; green,0; blue,64}, shape=circle]
\tikzstyle{new style 4}=[fill=blue, draw=blue, shape=circle]
\tikzstyle{new style 5}=[fill=white, draw=black, shape=circle, minimum size=0.6cm, inner sep=1pt]
\tikzstyle{new style 6}=[fill=magenta, draw=black, shape=rectangle, scale=1]
\tikzstyle{new style 7}=[fill={rgb,255: red,191; green,0; blue,64}, draw=black, shape=rectangle, scale=1]
\tikzstyle{error}=[fill=white, draw=black, shape=starburst, scale=0.5, starburst points=10, starburst point height=0.28cm]

\tikzstyle{red_dark}=[-, fill={rgb,255: red,250; green,129; blue,124}]
\tikzstyle{green_light}=[-, fill={rgb,255: red,240; green,254; blue,241}]
\tikzstyle{blue_light}=[-, fill={rgb,255: red,241; green,242; blue,255}]
\tikzstyle{red_light}=[-, fill={rgb,255: red,254; green,242; blue,241}]
\tikzstyle{green_dark}=[-, fill={rgb,255: red,183; green,214; blue,170}]
\tikzstyle{blue_dark}=[-, fill={rgb,255: red,168; green,194; blue,242}]
\tikzstyle{new edge style 0}=[-, fill={rgb,255: red,191; green,191; blue,191}]
\tikzstyle{new edge style 1}=[-, fill={rgb,255: red,250; green,129; blue,124}, line width=0.025cm, draw={rgb,255: red,220; green,105; blue,105}]